\documentclass[10pt,reqno]{amsart}
\usepackage{amsmath,amssymb}
\usepackage{exscale}
\usepackage{bbm}
\usepackage{tabularx,booktabs,longtable}
\usepackage{caption}

\usepackage{graphicx}
\newcommand{\Graph}[2][1.0]{\vcenter{\hbox{\includegraphics[scale=#1]{#2}}}}

\usepackage{mathtools}
\usepackage{numprint}
\newcommand{\lfrac}[2]{\frac{\numprint{#1}}{\numprint{#2}}}

\usepackage{color}
\definecolor{links}{rgb}{0,0.3,0}

\usepackage[colorlinks=true,citecolor=blue,urlcolor=links,breaklinks=true]{hyperref}

\usepackage[numbers,sort&compress]{natbib}

\usepackage[T1]{fontenc}
\usepackage{lmodern}
\usepackage{microtype}

\newcommand{\cored}{\mathop{{}^\backprime\!\Delta\!^\prime}}

\newcommand{\mzv}[2][]{\zeta^{#1}(#2)}
\newcommand{\zz}[1]{Z_{#1}}
\newcommand{\completed}[1]{\overline{#1}}
\DeclareMathOperator{\anc}{anc}
\DeclareMathOperator{\Spec}{Spec}
\newcommand{\non}[1]{\text{non-$\phi^{#1}$}}
\newcommand{\plogdiv}{primitive log-divergent}
\newcommand{\defas}{\mathrel{\mathop:}=}
\newcommand{\vect}[1]{\boldsymbol{#1}}
\newcommand{\iu}{\mathrm{i}}
\DeclareMathOperator{\ImTeil}{Im}
\DeclareMathOperator{\ReTeil}{Re}
\newcommand{\RU}{\xi}

\newcommand{\pipow}[1]{\nu_{#1}}

\DeclareMathOperator{\depth}{depth}

\newcommand{\isomorph}{\cong}
\newcommand{\conjugate}[1]{{#1}^{\ast}}
\usepackage{ifthen}
\newcounter{AxiomNumber}
\newcommand{\defaxiom}[1]{%
	\item[\textbf{I\arabic{AxiomNumber}:}]%
	\addtocounter{AxiomNumber}{-1}\refstepcounter{AxiomNumber}%
	\label{axiom:#1}\addtocounter{AxiomNumber}{1}%
}
\newcommand{\axiom}[1]{\textbf{\hyperref[axiom:#1]{I}\ref{axiom:#1}}}
\newcommand{\setexp}[2]{\left\{ #1 \colon #2 \right\}}
\newcommand{\set}[1]{\left\{ #1 \right\}}
\newcommand{\concat}{\star}

\newcommand{\abs}[1]{\left\vert#1\right\vert}
\newcommand{\restrict}[2]{\left.{#1}\right|_{#2}}
\newcommand{\MZV}[1][]{\mathcal{Z}_{#1}}
\newcommand{\MZVf}[1][]{\mathcal{U}_{\ifthenelse{\equal{#1}{}}{}{#1}}}
\newcommand{\MZVfDR}[1][]{\mathcal{U}^{\dR}_{\ifthenelse{\equal{#1}{}}{}{#1}}}

\newcommand{\dR}{\mathfrak{dr}}
\newcommand{\motivic}{\mathfrak{m}}
\DeclareMathOperator{\mot}{mot}
\newcommand{\motive}[1]{\mot_G}
\newcommand{\Periods}{\mathcal{P}}
\newcommand{\per}{\mathrm{per}}
\newcommand{\PhiPeriods}[1][]{\ifthenelse{\equal{#1}{}}{\Periods_{\phi^4}}{\Periods_{\phi^4,\leq #1}}}
\newcommand{\LogPeriods}[1][]{\ifthenelse{\equal{#1}{}}{\Periods_{\mathrm{log}}}{\Periods_{\mathrm{log},\leq #1}}}
\newcommand{\GraphPeriods}{\Periods}
\newcommand{\AllGraphPeriods}{\Periods_{0,0}}
\newcommand{\parity}[1]{\overline{#1}}
\newcommand{\DeligneBasis}[1][]{\mathcal{D}_{#1}}

\newcommand{\ParityBasis}[1][]{\mathcal{D}_{#1}'}

\newcommand{\divides}{\mathbin{|}}
\DeclareMathOperator{\Spann}{lin}
\newcommand{\lin}{\Spann}
\newcommand{\weight}[1]{\abs{#1}}

\newcommand{\Maple}{%
	\href{http://www.maplesoft.com/products/Maple/}
	{\textsf{\textup{Maple}}}%
}
\newcommand{\MapleNote}{%
	\footnote{Maple is a trademark of Waterloo Maple Inc.}
}
\newcommand{\MapleTM}{%
	\href{http://www.maplesoft.com/products/Maple/}
	{\textsf{\textup{Maple}}\texttrademark}%
}
\newcommand{\JaxoDraw}{\texttt{\textup{JaxoDraw}}}
\newcommand{\Axodraw}{\texttt{\textup{Axodraw}}}

\theoremstyle{plain}
\newtheorem{thm}{Theorem}
\newtheorem{lem}[thm]{Lemma}
\newtheorem{con}[thm]{Conjecture}
\newtheorem{cor}[thm]{Corollary}
\newtheorem{prop}[thm]{Proposition}
\newtheorem{quest}[thm]{Question}
\newtheorem{remark}[thm]{Remark}
\newtheorem{defn}[thm]{Definition}
\newtheorem{ex}[thm]{Example}

\numberwithin{thm}{section}
\numberwithin{equation}{section}

\newcommand{\perms}[1]{\mathfrak{S}_{#1}}
\newcommand{\shuffles}[1]{\mathfrak{S}_{#1}}

\DeclareMathOperator{\Li}{Li}
\newcommand{\dd}{\mathrm{d}}
\font \rus= wncyr10
\newcommand{\shuffle}{\mathbin{\hbox{\rus x}}}
\newcommand{\To}{\longrightarrow}
\newcommand{\sW}{{\mathcal W}}
\newcommand{\A}{{\mathbb A}}
\newcommand{\CC}{{\mathbbm{C}}}
\newcommand{\FF}{{\mathbb F}}
\newcommand{\NN}{{\mathbbm{N}}}
\newcommand{\PP}{{\mathbb P}}
\newcommand{\QQ}{{\mathbbm{Q}}}
\newcommand{\RR}{{\mathbbm{R}}}
\newcommand{\ZZ}{{\mathbbm{Z}}}

\title{The Galois coaction on $\phi^4$ periods}
\author{Erik Panzer}
\address{%
All Souls College,
OX1 4AL
Oxford, UK%
}
\email{erik.panzer@all-souls.ox.ac.uk}
\author{Oliver Schnetz}
\address{%
Department Mathematik\\
Cauerstra{\ss}e 11\\
91058 Erlangen}
\email{schnetz@mi.uni-erlangen.de}

\begin{document}
\begin{abstract}
We report on calculations of Feynman periods of {\plogdiv} $\phi^4$ graphs up to eleven loops.
The structure of $\phi^4$ periods is described by a series of conjectures. In particular, we discuss the possibility
that $\phi^4$ periods are a comodule under the Galois coaction. Finally, we compare the results with
the periods of {\plogdiv} \emph{non}-$\phi^4$ graphs up to eight loops and find remarkable differences
to $\phi^4$ periods.
Explicit results for all periods we could compute are provided in ancillary files.
\end{abstract}
\maketitle
\section{Introduction}
\subsection{Feynman periods}
Let $G$ be a connected graph. The graph polynomial of $G$ is defined by associating a variable $x_e$ to every edge $e$ of $G$ and setting (see \cite{Brown:PeriodsFeynmanIntegrals,BognerWeinzierl:GraphPolynomials})
\begin{equation}\label{eq:psi}
	\Psi_G(x)
	\defas
	\sum_{T\ \text{span.\ tree}}\;\prod_{e\notin T}x_e,
\end{equation}
where the sum is over all spanning trees $T$ of $G$. 
We say that $G$ is a $\phi^4$ graph when all its vertices have degree at most four. Moreover, a graph $G$ is called \plogdiv\ if
\begin{equation}\label{eq:defnprimdiv}\begin{split}
	N_G &= 2h_G\quad\text{and}\\
	N_{\gamma} &> 2h_{\gamma} 
	\quad \text{for all non-empty strict subgraphs $\gamma$ of $G$,}
\end{split}\end{equation}
where $h_{\gamma}$ denotes the `loop order' (first Betti number) of $\gamma$ and $N_{\gamma}$ is the number of edges in $\gamma$.
Whenever condition \eqref{eq:defnprimdiv} holds, the \emph{Feynman period} of the graph $G$ (not necessarily $\phi^4$) is defined by the convergent integral \cite{Weinberg,LowensteinZimmermann:PowerCountingMassless,BEK}
\begin{equation}\label{eq:Feynman-period}
	P(G) 
	= \int_0^\infty\!\!\!\cdots\int_0^\infty\frac{{\dd}x_1\!\cdots{\dd}x_{N_G-1}}{\restrict{\Psi_G(x)^2}{x_{N_G}=1}}
	\in \RR_+. 
\end{equation}
In this way, $P$ defines a map from the set of {\plogdiv} graphs to positive real numbers. 
In $\phi^4$ quantum field theory these numbers are renormalization scheme independent contributions to the $\beta$-function \cite{ItzyksonZuber}. 
Let
\begin{equation*}\begin{split}
	\PhiPeriods 
	&\defas \lin_{\QQ} \setexp{P(G)}{\text{$G$ {\plogdiv} and $\phi^4$}}\subseteq\\
	\LogPeriods 
	&\defas \lin_{\QQ} \setexp{P(G)}{\text{$G$ {\plogdiv}}}\subset \Periods
\end{split}\end{equation*}
denote the $\QQ$-vector spaces spanned by {\plogdiv} periods. They are subspaces of the $\QQ$-algebra $\Periods$ of periods
in the sense of Kontsevich and Zagier \cite{KZ:Periods}.\footnote{%
	Periods are those numbers which can be written as integrals of rational functions over domains defined by polynomial inequalities (with all coefficients in $\QQ$).
}
We obtain finite-dimensional subspaces if we restrict the loop order of the graphs,
\begin{equation}\label{eq:period-loop-grading}
	\mathcal{P}_{\bullet,\leq n}
	\defas \lin_{\QQ} \setexp{P(G)}{h_G\leq n}
\end{equation}
for $\phi^4$ graphs or general {\plogdiv} graphs $G$, respectively.

The first systematic study of $\phi^4$ periods (up to seven loops) was done with exact numerical methods in 1995 by D.\ Broadhurst and D.\ Kreimer \cite{BK:KnotsAndNumbers,BK:MZVPositiveKnots}.
Later, this method was extended to eight loops \cite{Schnetz:Census}.

The only all order result in $\phi^4$ theory is that periods of the zig-zag graphs (see Figure~\ref{fig:zig-zags}) are rational multiples of the Riemann zeta function at odd integers.
\begin{figure}%
	\centering%
	$\zz{5}=\Graph[0.9]{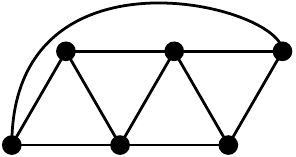}$ \qquad $\zz{6}=\Graph[0.9]{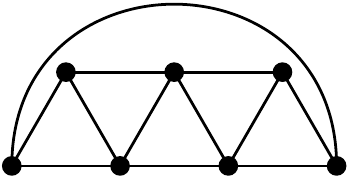}$%
	\caption{The zig-zag graphs with five and with six loops.}%
	\label{fig:zig-zags}%
\end{figure}
\begin{thm}[F.~Brown and O.~Schnetz \cite{BrownSchnetz:ZigZag}, conjectured in \cite{BK:KnotsAndNumbers}]
\label{thm:zig-zag}%
For every integer $n \geq 3$, the period of the zig-zag graph $\zz{n}$ is given by
\begin{equation} \label{eq:zigzag}
	P(\zz{n}) 
	=  4\frac{(2n-2)!}{n!(n-1)!}\Big(1-\frac{1-(-1)^n}{2^{2n-3}}\Big)\mzv{2n-3}.
\end{equation}
\end{thm}
The zig-zag periods are the only known $\phi^4$ periods which are rational multiples of values of the Riemann zeta function at integer argument.
A large class of $\phi^4$ periods (but not all $\phi^4$ periods) evaluate to \emph{multiple} zeta values (MZVs), i.e.\ rational linear combinations of the integer-indexed ($d,n_1,\ldots,n_d \in \NN=\set{1,2,3,\ldots}$) nested sums
\begin{equation}\label{eq:MZV}
	\mzv{\vect{n}}
	=
	\mzv{n_d,\ldots,n_2,n_1}
	\defas
	\sum_{k_d>\cdots>k_1\geq1}\frac{1}{k_d^{n_d}\!\cdots k_2^{n_2}k_1^{n_1}}
	\quad\text{with}\quad n_d\geq 2.
\end{equation}
We refer to $d$ as the \emph{depth} of the sum while $n_1+n_2+\ldots+n_d$ is called its \emph{weight}. The sums~\eqref{eq:MZV} obey many $\QQ$-linear relations and span the $\QQ$-algebra\footnote{There exist two weight-homogeneous product formulas in $\MZV$, see \cite{IharaKanekoZagier:DerivationDoubleShuffle} for example.}
\begin{equation}
	\MZV
	=\lin_{\QQ}\setexp{\mzv{n_d,\ldots,n_1}}{n_d\geq2}
\end{equation}
of MZVs. Conjecturally, MZVs are graded by the weight.

Feynman periods, and $\PhiPeriods$ in particular, are interesting because they are very sparse and appear to be constrained in several non-trivial ways. For example, no Feynman graph is known to evaluate to the simplest zeta value of all:
\begin{con}\label{con:zeta2}%
	$\mzv{2}$ is not a $\phi^4$ period: $\mzv{2} \notin \PhiPeriods$.\footnote{In fact, we expect $\mzv{2} \notin \LogPeriods$ and even more generally, that $\mzv{2} \notin \AllGraphPeriods$, see Section~\ref{sec:Feynman-motives}.}
\end{con}
This observation is well-known to particle physicists and supported by our explicit computations (all periods of graphs with loop order $\leq 7$ are known and MZVs of small weight are not expected to arise from larger graphs according to Conjecture~\ref{con:weightdrop}). Amazingly, a motivic (see next subsection) version of Conjecture~\ref{con:zeta2} recently became accessible by F.~Brown's `small graphs principle' in \cite{Brown:FeynmanAmplitudesGalois}, which was already used by Brown to prove the motivic version of $\log(2)\notin\PhiPeriods$ (see Theorem~\ref{thm:small-graphs-principle}).

However, $\phi^4$ periods are very sparse and $\MZV\cap\PhiPeriods$ appears to be a small subspace of $\MZV$, restricted much further than just by Conjecture~\ref{con:zeta2} alone.
Below we will also report on results for Feynman periods which (conjecturally) do not belong to $\MZV$, but still obey highly non-trivial constraints. 

The aim of this article is to offer an explanation (Conjecture~\ref{con:coaction}) of a wide range of phenomena (such as $\mzv{2}\mzv{2n+1}\notin\PhiPeriods$) that follow from this sparsity of $\phi^4$ periods.

\subsection{Galois theory for periods}
The most remarkable property of periods is that there should exist a Galois theory of periods \cite{Andre:GaloisMotivesTranscendental}, extending the classical Galois theory of algebraic numbers to the bigger space $\Periods$.
This construction rests on standard, but very difficult transcendence conjectures.\footnote{For example, $\pi$ and the numbers $\mzv{2n+1}$ are expected to be algebraically independent over the rationals, but $\mzv{3}$ is the only odd Riemann zeta value which is known to be irrational \cite{Apery}.}
To bypass this problem, one can define a $\QQ$-algebra $\Periods^{\motivic}$ of \emph{motivic} periods, which are enriched avatars of period integrals like \eqref{eq:Feynman-period}, see \cite{Brown:NotesMotivicPeriods}.
It comes with a surjective homomorphism $\per\colon \Periods^{\motivic} \longrightarrow \Periods$ whose injectivity is the remaining conjecture.
The powerful gain is a well-defined coaction\footnote{%
	Regrettably, our convention to write it as a left coaction is opposite to the notation as a right coaction, $\Delta\colon \Periods^{\motivic} \longrightarrow \Periods^{\motivic} \otimes \Periods^{\dR}$, which is used in \cite{Brown:NotesMotivicPeriods,Brown:FeynmanAmplitudesGalois}.
	This was noticed too late and changing our convention in this article would have been too risky.
	We hope that the reader will find this (purely notational) inconvience not too confusing when moving between our article and \cite{Brown:NotesMotivicPeriods,Brown:FeynmanAmplitudesGalois}.
}
\begin{equation}\label{eq:coact}%
	\Delta\colon \Periods^{\motivic} \longrightarrow \Periods^\dR \otimes  \Periods^\motivic
\end{equation}
which turns $\Periods^{\motivic}$ into a comodule over the Hopf algebra $\Periods^\dR$ of de Rham periods \cite{Brown:NotesMotivicPeriods}. The Galois group is dual to $\Periods^{\dR}$ and so \eqref{eq:coact} encodes the action of this group on $\Periods^{\motivic}$.
For the motivic multiple zeta values $\mzv[\motivic]{\vect{n}}$ defined in \cite{Brown:MMZ} this coaction can be computed explicitly, see \eqref{eq:Goncoact}.
The interpretation of the examples ($n\geq 1$)
\begin{equation}\label{eq:coact-singlezeta}%
	\Delta\mzv[\motivic]{2} = 1\otimes \mzv[\motivic]{2}
	\quad\text{and}\quad
	\Delta\mzv[\motivic]{2n+1}=\mzv[\dR]{2n+1}\otimes1+1\otimes\mzv[\motivic]{2n+1}
\end{equation}
is that this action is trivial on $\mzv[\motivic]{2}$, but $\mzv[\motivic]{2n+1}$ has $1$ as a (non-trivial) Galois conjugate.
We remove the trivial term $1 \otimes x$ from $\Delta x$ in the reduced coaction
\begin{equation}\label{eq:coact-reduced}%
	\Delta'x \defas \Delta x-1\otimes x,
\end{equation}
such that $\Delta'\mzv[\motivic]{2}=0$ and $\Delta'\mzv[\motivic]{2n+1} = \mzv[\dR]{2n+1} \otimes 1$.
The main subject of this article is to study the coaction on $\PhiPeriods^{\motivic}$, motivic versions of Feynman periods (see Section~\ref{sec:Feynman-motives}).
Surprisingly, the entirety of our data supports the following conjecture.
\begin{con}\label{con:coaction}%
The Galois coaction closes on $\phi^4$ periods:
\begin{equation*}
	\Delta\colon\PhiPeriods^{\motivic}\longrightarrow\Periods^\dR \otimes\PhiPeriods^{\motivic},
	\tag{Scenario 1}%
\end{equation*}
which means that all Galois conjugates of a $\phi^4$ period are also $\phi^4$ periods. Also, the data is consistent with the grading \eqref{eq:period-loop-grading} by loop order, if one allows for non-$\phi^4$ graphs:
\begin{equation*}
	\Delta'\colon \PhiPeriods[n]^{\motivic}\longrightarrow\Periods^\dR \otimes\LogPeriods[n-1]^{\motivic}.
	\tag{Scenario 2}%
\end{equation*}
\end{con}
By \eqref{eq:coact-singlezeta} the only non-trivial conjugate of odd zeta values is $1=P(\Graph[0.19]{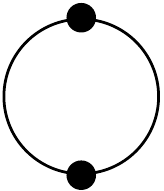})$.
Therefore the zig-zag series from Theorem~\ref{thm:zig-zag} is trivially consistent with both scenarios.
However, we already obtain strong constraints on $\PhiPeriods^{\motivic}$ if we assume the motivic version of Conjecture~\ref{con:zeta2}: Because $\mzv[\dR]{2n+1} \neq 0$ and
\begin{equation*}
	\Delta' \left( \mzv[\motivic]{2}\mzv[\motivic]{2n+1} \right)
	= \mzv[\dR]{2n+1} \otimes \mzv[\motivic]{2},
\end{equation*}
Scenario~1 implies that $\mzv[\motivic]{2}\mzv[\motivic]{2n+1} \notin \PhiPeriods^{\motivic}$ for every $n\geq 1$. In Section~\ref{sec:results} we give striking examples of non-MZV $\phi^4$ periods that obey the highly restrictive Conjecture~\ref{con:coaction}.
In Table~\ref{tab:periods} we show the Galois conjugates of all known $\phi^4$ periods with $\leq 8$ loops, expressed as linear combinations of $\phi^4$ periods. 
{\MapleTM} readable files with all known data are attached to this article, see Section~\ref{sec:data}.\MapleNote
The labeling of periods refers to these files and is consistent with \cite{Schnetz:Census}.

Our data rules out the possibility $\Delta'\colon \PhiPeriods[n]^{\motivic}\longrightarrow\Periods^\dR \otimes\PhiPeriods[n-1]^{\motivic}$, because some $8$-loop periods, like $P_{8,17}$, have Galois conjugates which are not in $\PhiPeriods[7]$ but require the $8$-loop period $P_{8,16}$ (see Table~\ref{tab:periods}).
Similarly, $\Delta'\colon \LogPeriods[n]^{\motivic}\longrightarrow\Periods^\dR \otimes\LogPeriods[n-1]^{\motivic}$ is excluded by \eqref{delta1Euler}.

\subsection{Motivic Feynman periods}\label{sec:Feynman-motives}
A construction of \emph{motivic} Feynman periods goes back to \cite{BEK,Brown:PeriodsFeynmanIntegrals} and was considerably generalized in \cite{Brown:FeynmanAmplitudesGalois}. To a graph $G$ one can associate a finite dimensional $\QQ$-vector space $\GraphPeriods^{\motivic}(G)$, see \cite[in particular section~9]{Brown:FeynmanAmplitudesGalois}, which consists of (motivic versions of) all integrals of the form
\begin{equation}\label{eq:generalized-periods}%
	\int_0^\infty\!\!\!\cdots\int_0^\infty \restrict{\frac{q(x)}{\Psi_G(x)^k}}{x_{N_G}=1}
	{\dd}x_1\!\cdots{\dd}x_{N_G-1}
	,
\end{equation}
where $k\geq 1$ is an integer and $q(x)\in\ZZ[x_1,\ldots,x_{N_G-1}]$ can be any polynomial such that the integral converges.%
\footnote{There is a caveat here: It is not clear that all elements of $\GraphPeriods^{\motivic}(G)$ can be written in the form \eqref{eq:generalized-periods}, see \cite{Brown:FeynmanAmplitudesGalois} for details on these \emph{non-global} periods.}
In this setting, the coaction conjecture is proven:
\begin{equation}\label{eq:coaction-theorem}%
	\Delta \colon 
	\GraphPeriods^{\motivic}(G)
	\longrightarrow 
	\Periods^{\dR} \otimes \GraphPeriods^{\motivic}(G)
	,
\end{equation}
in other words, the Galois group acts on $\GraphPeriods^{\motivic}(G)$. Following the notation of \cite{Brown:FeynmanAmplitudesGalois},
we write $\AllGraphPeriods^{\motivic}$ for the space spanned by all $\GraphPeriods^{\motivic}(G)$ for any\footnote{%
	A restriction to $\phi^4$ graphs and/or {\plogdiv} graphs  is insignificant here, because every graph is a minor of some {\plogdiv} $\phi^4$ graph \cite[Lemma~11]{Brown:PeriodsFeynmanIntegrals} and the spaces $\GraphPeriods^{\motivic}(G)$ are minor monotone by \cite[Theorem~7.8]{Brown:FeynmanAmplitudesGalois} (see \cite[Proposition~37]{Brown:PeriodsFeynmanIntegrals} for a non-motivic proof).}
graph $G$. Then \eqref{eq:coaction-theorem} partly explains Conjecture~\ref{con:coaction} because surprisingly, experiments suggest that at low loop orders $\AllGraphPeriods^{\motivic}$ and $\PhiPeriods^{\motivic}$ are the same.

However, this situation changes at high loop orders: At eight loops we have periods of non-$\phi^4$ graphs (e.g.\ $P^\non4_{8,39}$, see \eqref{delta1Euler}, and $P^\non4_{8,218}$) which are not expected in $\phi^4$. More evidence comes from the $c_2$-invariant which seems strongly constrained for $\phi^4$ periods (see Section~\ref{sec:c2}). In general, we expect that $\PhiPeriods^{\motivic}$ is tiny in $\LogPeriods^{\motivic}$ and hence also in $\AllGraphPeriods^{\motivic}$.
Regretfully, we have so far no information about how $\LogPeriods^{\motivic}$ relates to $\AllGraphPeriods^{\motivic}$ (except for inclusion).

\subsection{The $c_2$-invariant}\label{sec:c2}%
Since \eqref{eq:psi} is defined over the integers, it defines an affine scheme of finite type over $\Spec\ZZ$ which is called the graph hypersurface $X_G \subset \A^{N_G}$.
For any field $k$, we can therefore consider the zero locus $X_G(k)$ of $\Psi_G$ in $k^{N_G}$.
Using the finite fields $k\cong \FF_q$ of order $q$, this determines the point-counting function
\begin{equation}\label{2a}
	[X_G] \colon q \mapsto [X_G]_q \defas \# X_G(\FF_q) \in \NN_0
\end{equation}
as a map from the set of prime powers $q = p^n$ to non-negative integers. Inspired by the occurrence of MZVs in $\PhiPeriods$, Kontsevich raised the question if $[X_G]$ is a polynomial in $q$ \cite{Kontsevich:GelfandTalk1997}. While for graphs with at most 13 edges this is true \cite{Stembridge,Schnetz:Fq}, it is known that $[X_G]$ is of general type \cite{BelkaleBrosnan} and fails to be polynomial even for $\phi^4$ graphs \cite{Doryn:Pointcount,Schnetz:Fq,modphi4}.

For every graph with at least three vertices, $[X_G]_q$ is divisible by $q^2$ \cite{Schnetz:Fq}. This suggests
\begin{defn}\label{def:c2}
Let $G$ have at least three vertices. The $c_2$-invariant \cite{Schnetz:Fq} associates to $G$ an infinite sequence in $\prod_{q=p^n}\ZZ/q\ZZ$ by
\begin{equation}\label{2b}
	c_2(G)_q
	\defas
	[X_G]_q/q^2 \mod q.
\end{equation}
\end{defn}
By experiment \cite{Schnetz:Fq} it is expected that for many graphs the $c_2$-invariant is the most complicated part of the point-counting function $[X_G]$. Recent work \cite{Schnetz:Fq,K3phi4,BrownDoryn:Framings} showed that it captures some information about the period~\eqref{eq:Feynman-period}; in particular it is conjectured in \cite{K3phi4} that graphs with identical period have the same $c_2$-invariant. 

Following \cite{BelkaleBrosnan}, one might expect $c_2$-invariants of graphs to be arbitrarily complicated. 
Experiments up to ten loops in \cite{modphi4}, however, suggest that $c_2$-invariants of $\phi^4$ graphs are very restricted; in small dimensions we see only few geometries:

\subsubsection{quasi-constants}
Many $c_2$-invariants count points on zero-dimensional varieties. In this case they are called \emph{quasi-constant} because there exist an $m\in \NN$ and a unique $c\in\ZZ$ such that $c_2(G)_{p^{nm}} \equiv c \mod p^{nm}$ for all $n\in\NN$ and all but finitely many primes $p$.
We only know cases when $c=0$ and $c=-1$. In fact, so far only five quasi-constant $c_2$-invariants were found in $\phi^4$ theory:
The constants $0,-1$ and the three quasi-constants (with $c=-1$) $-z_2$, $-z_3$, and $-z_4$ which depend on whether $\FF_q$ has a primitive $N$th root of unity (this is the case $z_N(q)=1$) or not:
\begin{equation}
	z_N(q)
	=\begin{cases}
		\phantom{-}1 & \text{if $N \divides q-1$,}\\
		\phantom{-}0 & \text{if $\gcd(N,q)>1$, and}\\
		-1& \text{otherwise.}\\
	\end{cases}
\end{equation}
We say that a graph $G$ has weight drop if $c_2(G)\equiv0\mod q$ for all $q$. This property is linked to the weight of its period:
In the case that $P(G)$ is an MZV, its weight is at most $2h_G-3$. If $G$ has weight drop then it is conjectured (and proved in some cases \cite{BrownYeats:WD}) that $P(G)$ has at most weight $2h_G-4$.

The $c_2$-invariant $-1$ indicates an MZV period of pure maximum weight (Conjecture~\ref{con:c2-1}) whereas the periods of quasi-constant $c_2=-z_N$ are MPLs of pure maximum weight at $N$th roots of unity by Conjectures~\ref{con:c2-z2} and \ref{con:z3}.

Beyond $\phi^4$ theory we see more quasi-constant $c_2$-invariants: The non-$\phi^4$ graph $P^\non4_{8,218}$ e.g.\ has $c_2$-invariant
$c_2(P^\non4_{8,218})_q=1-[X\subset\A\colon x^2+x-1=0]_q$.

\subsubsection{modular forms}
In $\phi^4$ theory, non-quasi-constant $c_2$-invariants seem to come from geometries of dimensions greater than or equal to two. 
Only three two-dimensional geometries are found in $\phi^4$ theory up to ten loops.
These are of $K3$ type and modular with respect to modular forms of weight and level (3,7), (3,8), and (3,12) \cite{modphi4,K3phi4}.

The number of geometries of $\phi^4$ $c_2$-invariants seem to increase with the dimension. The first non-modular varieties are expected at nine loops and four dimensions \cite{modphi4}.
The $c_2$-invariant is (at least for small primes $p$) known for all $\phi^4$ periods up to ten loops and all non-$\phi^4$ periods up to nine loops \cite{modphi4}.

With present technology, periods of graphs with non-quasi-constant $c_2$-invariant cannot be calculated.
Hence, all periods we consider in this article come from graphs with quasi-constant $c_2$.

\subsection{Results}\label{sec:results}
In the subsequent analysis of $\phi^4$ periods we assume the standard transcendence conjecture that $\per\colon \Periods^{\motivic}\longrightarrow \Periods$ is an isomorphism. 
This means that we compute the real numbers $P(G)\in\Periods$ analytically in terms of, say, real MZVs and subsequently apply the coaction formula \eqref{eq:Goncoact} for motivic MZVs
(see Remark~\ref{rem:period-conjecture}).\footnote{This is well-defined only conjecturally, because there could be additional relations between MZVs which are not shared by motivic MZVs,
making the replacement $\mzv{\vect{n}} \mapsto \mzv[\motivic]{\vect{n}}$ ambiguous.}

We succeeded in calculating all 17 $\phi^4$ periods up to seven loops and 23 out of at most 31 $\phi^4$ periods at eight loops.
Beyond eight loops we were still able to calculate an increasing number of periods which however are out-counted by the quickly increasing number of all $\phi^4$ periods:
At nine loops we know 47 out of at most 134 periods, at ten loops we know 88 out of at most 846 periods and at eleven loops we know 125 out of at most 6300 periods.
The list of non-$\phi^4$ periods is also complete up to seven loops. At eight loops it was possible to calculate most non-$\phi^4$ periods.

\begin{figure}%
	\centering%
	\includegraphics[width=0.25\textwidth]{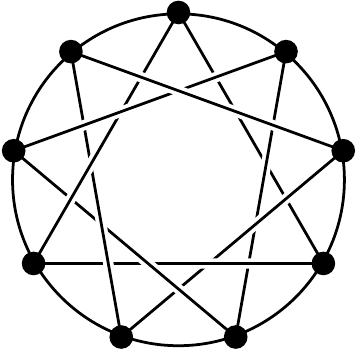}%
	\caption{The completion (see Definition~\ref{def:completion}) of the graph $P_{7,11}$ is the circulant $C^9_{1,3}$. Its period \eqref{eq:p711} evaluates to multiple polylogarithms at sixth roots of unity.}%
	\label{fig:P711}%
\end{figure}
One long sought-after period is $P_{7,11}$ in the nomenclature of \cite{Schnetz:Census} (see Figure \ref{fig:P711}; we use the name $P_{7,11}$ for the graph as well as for its period).
It is the first (conjectured) non-MZV $\phi^4$ period. It can be expressed in terms of multiple polylogarithms (MPLs) at sixth roots of unity.
To make the coaction easily visible, we express our results in a suitable \emph{$f$-alphabet} of non-commutative words, a construction which we will recall in Section~\ref{sec:falphabet}. It allows us to represent periods with abstract words (elements in a tensor algebra) such that the coaction simply becomes deconcatenation.
For example, motivic MZVs correspond to words in letters of odd weight $\geq 3$, see \cite{Brown:DecompositionMotivicMZV}:
\begin{equation}\label{eq:falphabet-MZV}%
	\MZV^{\motivic} 
	\isomorph
	\QQ\langle f_3, f_5, f_7, \cdots \rangle\otimes\QQ[\pi^2]
	.
\end{equation}
Deligne gave bases for MPLs at $N$th roots of unity in the cases $N\in\set{2,3,4,6,8}$ in \cite{Deligne:GroupeFondamental} (see Section~\ref{sec:deligne}) which lead to $f$-alphabets for these periods. These alphabets have additional letters (compared to MZVs), for example $f_1$ or even letters $f_{2k}$.

Because many%
\footnote{%
	All $\phi^4$ periods of loop order at most six are MZVs. 
	At higher loop orders we expect the number of MZV periods to grow at least exponentially, although the ratio of MZV periods over all periods is expected to approach zero when the loop order goes to infinity.
}
$\phi^4$ periods are MZVs it is convenient to use an $f$-alphabet in which the subspace of MZVs is spanned by the words in letters of odd weight $\geq 3$ as in \eqref{eq:falphabet-MZV}. 
Deligne's original basis does not have this property. We prove in Theorem~\ref{thm:falphabet-parity} that taking real parts if the weight plus the depth of the basis element is even and $\iu \times$ imaginary parts otherwise leads to a basis that embeds MZVs in the natural way.

Using the $f$-alphabet (letters $f^6_2,$ $f^6_3,$ $f^6_4,\ldots$) with respect to this modified Deligne basis (see Remark~\ref{rem:MZVoverQ}) we found the following result for $P_{7,11}$:
\begin{equation}\label{eq:p711}\begin{split}
	P_{7,11}
	&=-\lfrac{332262}{43}f^6_8f^6_3
	+\lfrac{54918}{55}f^6_6f^6_5
	+\lfrac{1134}{13}f^6_4f^6_7
	-\lfrac{1874502}{3485}f^6_2f^6_9
	\\&\quad
	+\numprint{5670}f^6_2f^6_3f^6_3f^6_3
	-\tfrac{\numprint{3216912825399005402331281812377062149}}{\numprint{14080217073343074027422017273458000}}\Big(\frac{\pi}{\sqrt{3}}\Big)^{11}.
\end{split}\end{equation}
Note that the $f$-alphabet is not fully canonical, see e.g.\ \eqref{P711N3} for an alternative representation.
The presence of large numerators and denominators reflects a certain lack of economy of the modified Deligne basis for expressing Feynman periods
(the coefficients are even larger in the $f$-alphabet with respect to Deligne's original basis).
The construction of a better adapted basis is discussed in \cite{Broadhurst:Aufbau}.

Note that \eqref{eq:p711} is consistent with Scenarios~1 and 2: The new letters $f^6_2$, $f^6_4$, $f^6_6$, $f^6_8$ appear only at the left-most position. 
On the right-hand side of the tensor product in $\Delta'P_{7,11}$ (which gives the non-trivial Galois conjugates of $P_{7,11}$) only odd letters and thus MZVs appear
(this is the point of Theorem~\ref{thm:falphabet-parity}).
In fact, all these MZVs are in the subspace $\PhiPeriods[6]$. Equation \eqref{eq:p711} is also consistent with 

\begin{thm}[{small graphs principle, \cite[Section~9.3]{Brown:FeynmanAmplitudesGalois}}]\label{thm:small-graphs-principle}
	Let $\RU_6$ be a primitive sixth root of unity. Then $\Li_2^{\motivic}(\RU_6)$ is not a Galois conjugate of any period in $\PhiPeriods^{\motivic}$.
In other words, $f^6_2$ cannot appear as the right-most letter in the $f$-alphabet representation.
\end{thm}
Other examples of MPLs at sixth roots of unity are the eight and nine loop periods $P_{8,33}$ and $P_{9,136}=P_{9,149}$. All these periods have $c_2$-invariant $-z_3$ and their expressions in terms of the $f^6_\bullet$ look similar to \eqref{eq:p711} with very large numerators and denominators. They also obey Conjecture~\ref{con:coaction}. Explicit results are in the attached files.

Alternating sums (MPLs at $-1$) of weights $12$, $14$ and $15$ were found in $\phi^4$ periods of loop orders nine and ten, see Sections~\ref{sec:c20} and \ref{sec:z2}.
The nine loop weight drop period $P_{9,36}=P_{9,75}$ mixes weight $13$ (an MZV) and weight $12$. The weight $12$ part of $P_{9,36}=P_{9,75}$ is
\begin{equation}\label{eq:p93612}\begin{split}
	\sW_{12}P_{9,36}
	&=\sW_{12}P_{9,75}
	=p_{12}+Q_{12,4}
	\quad\text{where}\\
	Q_{12,4}
	&=\lfrac{8634368}{135}f^2_1f^2_5f^2_3f^2_3
	 -\lfrac{3899392}{135}f^2_1f^2_3f^2_5f^2_3
	 +\lfrac{458752}{27}f^2_1f^2_3f^2_3f^2_5
	\\&\quad
	 +\lfrac{222208}{38475}f^2_1f^2_3\pi^8
	 -\lfrac{71206701679851520}{59408350617}f^2_1f^2_{11}.
\end{split}\end{equation}
Here, $p_{12}\in\sW_{12}\PhiPeriods[8]$ is an MZV. Note that \eqref{eq:p93612} is consistent with Scenarios~1 and 2 and with the small graphs principle in complete analogy to \eqref{eq:p711}:
All non-trivial Galois conjugates are MZVs in $\PhiPeriods[8]^{\motivic}$ and $f_1^2$ never appears as right-most letter (a proven general restriction on $\phi^4$ periods which follows from
\cite[Theorem~9.4]{Brown:FeynmanAmplitudesGalois}).

Conjecturally, there exist no alternating sums of weight less than $12$ in $\PhiPeriods^{\motivic}$.

An alternating sum of weight 15 was found in $P_{9,67}$, a period with $c_2$-invariant $-z_2$. It is (conjecturally) one of the smallest graphs in $\phi^4$ with $c_2$-invariant $-z_2$ which
is not an MZV. One of its Galois conjugates is the weight 12 alternating sum
\begin{equation}\label{eq:Q12,5}\begin{split}
	Q_{12,5}
	&\defas 
	\lfrac{777728}{45}f^2_1f^2_5f^2_3f^2_3
	+\lfrac{990976}{45}f^2_1f^2_3f^2_5f^2_3
	-\lfrac{163072}{9}f^2_1f^2_3f^2_3f^2_5
	\\&\quad\ 
	-\lfrac{194432}{38475}f^2_1f^2_3\pi^8
	+\lfrac{9739832477359040}{25460721693}f^2_1f^2_{11}
\end{split}
\end{equation}
which is \emph{not} a rational multiple of $Q_{12,4}$ in \eqref{eq:p93612}. However, $Q_{12,5}$ was found in the weight 12 parts of several weight drop
\emph{non}-$\phi^4$ eight loop periods, e.g.\ in $P^\non4_{8,433}$.

So, for all known periods Scenario~2 is true. For Scenario~1 we miss $Q_{12,5}$ in the space of known $\phi^4$ periods.
Assuming Scenario~1 and combining all conjectures on the structure of $\phi^4$ periods (see Section~\ref{sec:phi4}) we can localize $Q_{12,5}$ in the period of a single ten loop graph: $P_{10,425}$, see Conjecture~\ref{con:L10}. Regretfully, the period of this graph cannot yet be calculated.
Scenario~1 holds for all known graphs of at most eight loops. This is explicit in Table~\ref{tab:periods} at the end of this article.
Text files with all known data are attached to this article and also described in Section~\ref{sec:data}.

Conjecturally, we know all mixed Tate periods in $\PhiPeriods$ up to weight $11$, see Section~\ref{sec:c20}. Assuming Scenario~1 we find the dimensions shown in Table~\ref{tab:dimensions}.
\begin{remark}\label{remark:hom}
Some (though only few) of the mixed Tate $\phi^4$ periods mix weights (see Conjecture~\ref{con:weightdrop}), so it is not clear that they span a weight homogeneous subspace of MPLs.
However, this turns out to be true for the subspace spanned by all periods that we were able to compute.
In Table~\ref{tab:dimensions} we assume that this holds for all mixed Tate $\phi^4$ periods of weight $\leq 12$ and we expect exceptions to weight homogeneity (if at all) only at high loop orders. 
Note that weight homogeneity is equivalent to stability under the action of the reductive part of the Galois group, which is isomorphic to $\QQ^{\times}$ and $\lambda\in\QQ^{\times}$ acts on a period of pure weight $w$ by multiplication with $\lambda^w$ \cite{Brown:NotesMotivicPeriods}.

Up to weight $12$, our data is also consistent with a $\QQ$-algebra structure of $\PhiPeriods$ (see Question~\ref{quest:alg}). This means that all products of known $\phi^4$ periods with total weight $\leq 12$ can be expressed as $\QQ$-linear combinations of known $\phi^4$ periods.
\end{remark}

\begin{table}
\begin{tabular}{lrrrrrrrrrrrrr}\toprule
weight&0&1&2&3&4&5&6&7&8&9&10&11&12\\\midrule
$\QQ-$dimension&1&0&0&1&0&1&1&1&2&2&3&5&8?\\
algebra generators&
1&0&0&1&0&1&0&1&1&1&1&3&5?\\
\bottomrule
\end{tabular}%
\caption{Conjectured dimensions of the mixed Tate subspace of $\PhiPeriods$ (periods expressible as MPLs) in different weights.{\protect\footnotemark} The last row is obtained by discarding products of lower-weight periods, like $\mzv[2]{3}$ which spans $\PhiPeriods$ in weight $6$. Details are given in Section~\ref{sec:c20}.}%
\label{tab:dimensions}%
\end{table}
\footnotetext{The weight is only a filtration on $\Periods^{\motivic}$, but becomes a grading in this very special case of polylogarithms. Note that our weights are half the Hodge-theoretic weights; see \cite{Brown:NotesMotivicPeriods} for details.}

In Section~\ref{sec:phi4} we give a series of conjectures on the structure of $\PhiPeriods$. By counter-examples we prove that some of these conjectures are false in $\LogPeriods$. This gives a remarkable difference between $\PhiPeriods$ and $\LogPeriods$. We conclude that there possibly exists a structure which is only present in $\PhiPeriods$. Mathematically it is not yet clear what mechanism distinguishes $\phi^4$ periods from periods of all graphs.
Finally we discuss potential counter-examples to Conjecture~\ref{con:coaction} in Sections~\ref{sec:L10} and \ref{sec:non-MPL}.

We close the introduction with the remark that the first three orders of the quantum electrodynamical contribution to the anomalous magnetic dipole moment of the electron are alternating sums which also have the word $f_1^2 f_3^2$ but not the word $f_3^2 f_1^2$ (in line with \eqref{eq:p93612}): The third order contribution to $(g-2)/2$ is given by \cite{LaportaRemiddi:g2alpha3}
\begin{equation}\label{eq:electron}%
	\lfrac{83}{72}f_3\pi^2-\lfrac{215}{24}f_5-\lfrac{350}{9}f^2_1f^2_3+\lfrac{511}{2160}\pi^4+\lfrac{139}{18}f_3+\lfrac{298}{9}f^2_1\pi^2+\lfrac{17101}{810}\pi^2+\lfrac{28259}{5184}.
\end{equation}
As the coaction theorem, Eq.~\eqref{eq:coaction-theorem}, also applies to Feynman periods with masses and momenta \cite{Brown:FeynmanAmplitudesGalois},
our findings related to the $\beta$-function of $\phi^4$ give hope that one might find similarly remarkable structures in other physical observables, such as $g-2$. 
Interestingly, recent computations \cite{AbreuBrittoGroenqvist:massive,AbreuBrittoDuhrGardi:cuts2cop} of the coaction of mass- and momentum dependent Feynman periods already led to a diagrammatic coaction formula for one-loop graphs \cite{Abreu:Amplitudes2015}.

\subsection*{Acknowledgements.}
Through various stages of this work, Erik Panzer was supported by Humboldt University, the CNRS via ERC grant 257638 and All Souls College, Oxford. 
This work has been done while Oliver Schnetz was visiting scientist at Humboldt University, Berlin.
Oliver Schnetz is supported by DFG grant SCHN~1240/2-1.
Both authors are grateful to Francis Brown for many extremely valuable discussions and comments on this draft. We also thank Cl\'{e}ment Dupont for pointing out the issues addressed in Section~\ref{sec:non-MPL} and Claire Glanois for helpful feedback on motivic MPLs at roots of unity.
Furthermore we are indebted to two very attentive referees for thoroughly reading an earlier version of this manuscript and giving many valuable suggestions.

Computations were performed on the computers beta, gamma, and delta of the Institute of Mathematics, Humboldt Universit\"{a}t zu Berlin using {\Maple} version 16.
Several figures were created using {\Axodraw} \cite{Vermaseren:Axodraw} and {\JaxoDraw} \cite{BinosiTheussl:JaxoDraw}.

\section{Iterated integrals and multiple zeta values}

\subsection{Iterated integrals}
In \cite{Chen} Chen developed a theory of iterated path integration on general manifolds.
Here, we need only the elementary one-dimensional case of a punctured sphere $\CC\setminus \Sigma$ for some finite set $\Sigma \subset \CC$.
Fix a weight $n\in\NN$, a path $\gamma\colon [0,1]\To\CC\setminus \Sigma$ from $a_0=\gamma(0)$ to $a_{n+1}=\gamma(1)$ and differential forms $\omega_i(z)\defas\dd z/(z-a_i)$ with $a_i \in \Sigma$ for $i\in\set{1,2,\ldots,n}$. 
Then, the iterated integral of the word $a_n\!\ldots a_1$ (considering $\Sigma$ as an alphabet) along $\gamma$ is defined by
\begin{equation}
	I(a_{n+1},a_n\ldots a_1,a_0)_\gamma
	=
	\int_{1>t_n>\cdots>t_1>0}\gamma^\ast\omega_n(t_n)\wedge\ldots\wedge\gamma^\ast\omega_1(t_1),
\end{equation}
where the integration simplex is endowed with the standard orientation and $\gamma^\ast\omega_i$ is the pullback of $\omega_i$ along $\gamma$.
Iterated path integrals have the following properties:
\begin{itemize}
	\defaxiom{empty}
		$I(a_1,a_0)_\gamma=1$ (by definition).
	\defaxiom{homotopy}
		$I(a_{n+1},w,a_0)_\gamma$ depends only on the homotopy class of $\gamma$.

	\defaxiom{constant}
		$I(a_0,w,a_0)_\gamma=0$ for the constant path $\gamma(t)=a_0$ and non-empty words $w$.

	\defaxiom{reverse}
		Path reversal: For the reversed path $\gamma^{-1}(t) = \gamma(1-t)$,
		\begin{equation*}
			I(a_{n+1},a_n\ldots a_1,a_0)_{\gamma^{-1}}
			=
			(-1)^nI(a_0,a_1\ldots a_n,a_{n+1})_\gamma
			.
		\end{equation*}
	
	\defaxiom{concat}
		Path concatenation:
	If $\gamma=\gamma_2 \concat \gamma_1$ is the composition of (first) $\gamma_1$ and (second) $\gamma_2$ meeting at $x=\gamma_1(1)=\gamma_2(0)\in\CC\setminus \Sigma$, then
	\begin{equation*}
		I(a_{n+1},a_n\ldots a_1,a_0)_\gamma
		=
		\sum_{k=0}^nI(a_{n+1},a_n\ldots a_{k+1},x)_{\gamma_2}I(x,a_k\ldots a_1,a_0)_{\gamma_1}
		.
	\end{equation*}

	\defaxiom{shuffle}
	Shuffle product:
	For $n=r+s$ let $\shuffles{r,s}\subset \perms{n}$ denote the $(r,s)$-shuffles
	\begin{equation*}
		\shuffles{r,s}
		=\setexp{\sigma\in\perms{n}}{%
			\sigma^{-1}(1)<\ldots<\sigma^{-1}(r)
			\ \text{and}\ 
			\sigma^{-1}(r+1)<\ldots<\sigma^{-1}(n)
		},
	\end{equation*}
	a subset of the group $\perms{n}$ of permutations of $\set{1,\ldots,n}$. Then,
	\begin{equation*}
		I(a_{n+1},a_n\ldots a_{r+1},a_0)_\gamma I(a_{n+1},a_r\ldots a_1,a_0)_\gamma
		=
		\sum_{\mathclap{\sigma\in\shuffles{r,s}}}I(a_{n+1},a_{\sigma(n)}\ldots a_{\sigma(1)},a_0)_\gamma
		.
	\end{equation*}

	\defaxiom{moebius}
	Chain rule: Every M\"{o}bius transformation $f$ maps a word $w$ by pull-back of the differential forms $\omega_i$ to a linear combination $f^\ast w$ of words in the alphabet $f^{-1}(\Sigma \cup \set{\infty})\setminus\set{\infty}$. The iterated integral transforms as
	\begin{equation*}
		I(f(a_{n+1}),w,f(a_0))_{f(\gamma)}
		=
		I(a_{n+1},f^\ast w,a_0)_\gamma
		,
	\end{equation*}
	where the iterated integral on the right-hand side is linearly extended to linear combination of words.
\end{itemize}
Although by definition one assumes $a_0$, $a_{n+1}\notin \Sigma$, it is possible by a limiting procedure to define iterated integrals for the singular cases $a_0$, $a_{n+1}\in \Sigma$.
See \cite{Panzer:PhD} for a detailed discussion of iterated integrals in this special setup.

We only consider straight paths and suppress the subscript $\gamma$.
As a special case we obtain the multiple polylogarithms (MPLs) from \cite{Goncharov:PolylogsArithmeticGeometry} in the form
\begin{equation}\label{eq:Li}
	(-1)^r I(z,\underbrace{0\ldots01}_{n_r}\ldots\underbrace{0\ldots01}_{n_1},0)
	= \Li_{n_r,\ldots, n_1}(z)
	= \sum_{k_r>\cdots>k_1\geq1}\frac{z^{k_r}}{k_r^{n_r}\!\cdots k_1^{n_1}}.
\end{equation}

By \eqref{eq:MZV}, iterated integrals over the letters $\Sigma=\set{0,1}$ are MZVs. In this article we also need MPLs at other roots of unity. 
In particular, we consider the three-letter alphabet $\Sigma = \set{0,1,\exp(2\pi\iu/N)}$ for $N\in\set{1,2,3,4,6}$. In the case of alternating sums $N=2$ a basis was first conjectured by D.~Broadhurst \cite{Broadhurst:IrreducibleEulerSums}. The more general case was treated by P.~Deligne in \cite{Deligne:GroupeFondamental}.

\subsection{Motivic iterated integrals}
One can define motivic iterated integrals $I^{\motivic}(\vect{a})$ which share the above axioms and evaluate to the ordinary iterated integrals under $\per\left( I^{\motivic}(\vect{a}) \right) = I(\vect{a})$, see \cite{Brown:SingleValued,Brown:MMZ,Glanois:Descents} for details.
While questions of transcendence and $\QQ$-linear independence are extremely hard for ordinary iterated integrals, they are (for some geometries) perfectly understood and proven in the motivic setup.
In particular, it is only possible to prove dimension formulae like Corollary~\ref{cor:dims} for motivic objects.
Moreover, strictly speaking, the coaction exists only for motivic iterated integrals. 
In order to get a well-defined coaction on ordinary iterated integrals, we need the conjectured injectivity of the evaluation map $\per$.
In the following we (tacitly) work with motivic iterated integrals.

A purely combinatorial formula for the Galois coaction on (motivic) iterated integrals was found by A.~Goncharov \cite{Goncharov:FundamentalGroupoids} in the de Rham setting (where $2\pi\iu=0$).
It was proved by F.~Brown in \cite{Brown:MMZ} that it extends to motivic iterated integrals:
\begin{align}\label{eq:Goncoact}
	&\Delta I(a_{n+1},a_n\!\ldots a_1,a_0)
	= 
	\\&
	\sum_{k=0}^n\;\sum_{i_{k+1}>\cdots>i_1>i_0}
		\prod_{p=0}^k I^{\dR}(a_{i_{p+1}},a_{i_{p+1}-1}\ldots a_{i_p+1},a_{i_p})
		\otimes I(a_{i_{k+1}},a_{i_k}\!\ldots a_{i_1},a_{i_0}),
	\nonumber%
\end{align}
where the sum is over indices satisfying $i_0=0$ and $i_{k+1}=n+1$. 
On the left-hand side of the tensor product, $I^\dR(\ldots)$ is the image of the projection of $I(\ldots)$ onto the factor algebra of iterated integrals modulo the ideal generated by $2\pi\iu$.\footnote{It is not clear if this picture of de Rham periods generalizes to more complicated Feynman graphs. According to \cite[Section~4.3]{Brown:NotesMotivicPeriods}, such a representation via a projection of motivic periods works only for \emph{separated} graphs \cite[Section~9.2]{Brown:FeynmanAmplitudesGalois}. It is yet unknown which graphs have this property.}

\begin{ex}\label{ex:53}
Consider $I(1,00001001,0)=\mzv{5,3}$ according to \eqref{eq:Li}.
The terms with $k=0$ and $k=n$ in the coaction are trivial, so that
\begin{equation*}
	\Delta\mzv{5,3}
	=1\otimes\mzv{5,3}+\mzv[\dR]{5,3} \otimes 1+\cored\mzv{5,3},
\end{equation*}
where $\cored\mzv{5,3}$ contains the summands with $1\leq k \leq 7$.
Note that iterated integrals of words with only one type of letters vanish: For all $m\geq 1$ and arbitrary $j$ and $k$,
\begin{equation}\label{eq:vanishing-I}
	I(a_j,\underbrace{0\ldots0}_{m},a_k)
	= I(a_j,\underbrace{1\ldots1}_{m},a_k)
	=0.
\end{equation}
This follows from axiom \axiom{constant} when $a_j=a_k$ and otherwise from $I(1,0\ldots0,0)=(1/m!)I(1,0,0)=0$ via \axiom{shuffle} and regularization (using \axiom{reverse} and \axiom{moebius} with $f(z)=1-z$ this extends to all combinations of $0\ldots0$ and $1\ldots1$ with $a_j,a_k\in\set{0,1}$).
Thus we must have $k\geq 2$ and the letters $a_{i_1},\ldots,a_{i_k}$ must contain at least one $0$ and at least one $1$ for $I(a_{i_{k+1}},a_{i_k}\!\ldots a_{i_1},a_{i_0})$ to be non-zero. This leaves at most one letter $1$ for the de Rham side, so only one of the de Rham factors can have a letter $1$. For the product to be non-zero, \eqref{eq:vanishing-I} dictates that all other factors must have the form $I^{\dR}(a_{i_{p+1}},a_{i_p})=1$. So in our case of depth 2, the coaction~\eqref{eq:Goncoact} boils down to
\begin{equation*}
	\cored I(a_{n+1},a_n\!\ldots a_1,a_0)
	= \sum_{\mathclap{0\leq j<k-1\leq n}} 
		I^{\dR}(a_{k},a_{k-1}\ldots a_{j+1},a_j)
		\otimes
		I(a_{n+1},a_n\!\ldots a_k a_j \ldots a_1,a_0)
	.
\end{equation*}
The only summands in this formula for $\mzv{5,3}$ that do not vanish due to \eqref{eq:vanishing-I} are $(j,k)=(0,4)$, $j=1$ with $5\leq k \leq 8$, and $j=2,3$ with $k=9$. Explicitly, this yields
\begin{equation}\begin{split}\label{eq:cored-zeta(5,3)}
	\cored \mzv{5,3}&=
	(I^{\dR}(0,100,1)+I^{\dR}(1,001,0)) \otimes I(1,00001,0)
	\\&+\phantom{(}I^{\dR}(0,0100,1) \otimes I(1,0001,0)
	\\&+(I^{\dR}(1,00001,0)+I^{\dR}(0,00100,1)) \otimes I(1,001,0)
	\\&+(I^{\dR}(1,000010,0)+I^{\dR}(0,000100,1)) \otimes I(1,01,0).
\end{split}\end{equation}
All of these iterated integrals are proportional to Riemann zeta values by \axiom{reverse} and
\begin{equation*}
	I(1,\underbrace{0\ldots0}_{n-m}1\underbrace{0\ldots0}_m,0)
	= (-1)^m \binom{n}{m} I(1,\underbrace{0\ldots0}_{n}1,0)
	= (-1)^{m+1} \binom{n}{m} \mzv{n+1},
\end{equation*}
which follows from \axiom{shuffle} and the regularization $I(1,0,0)=0$. For example,
\begin{align*}
	0 &= I(1,0,0) I(1,00001,0)
	   = 5 I(1,000001,0) + I(1,000010,0)
	\\
	  &= -5\mzv{6} + I(1,000010,0)
	  .
\end{align*}
After we rewrite \eqref{eq:cored-zeta(5,3)} using \axiom{reverse} and the formula from above, we get
\begin{equation*}
	\cored \mzv{5,3}
	= 3 \mzv[\dR]{4} \otimes \mzv{4}
	-5\mzv[\dR]{5} \otimes \mzv{3}
	-15 \mzv[\dR]{6} \otimes \mzv{2}
	=-5\mzv[\dR]{5}\otimes \mzv{3}
\end{equation*}
since $\mzv{4}=\pi^4/90$ and $\mzv{6}=\pi^6/945$ vanish in the de Rham quotient by $(2\pi\iu)^{\dR}=0$.
\end{ex}

\subsection{The Deligne basis}\label{sec:deligne}
In \cite{Brown:MMZ} F.~Brown proves that MZVs in 2s and 3s form a basis of (motivic) MZVs (the Hoffman basis).
In 1996 D.~Broadhurst conjectured a basis for alternating sums \cite{Broadhurst:IrreducibleEulerSums}. This conjecture was first proved by P.~Deligne \cite{Deligne:GroupeFondamental} (via a study of the motivic fundamental groupoid of $\PP^1\setminus\set{0,\pm 1,\infty}$ and showing that it generates the Tannakian category of mixed Tate motives over $\ZZ[\tfrac{1}{2}]$), who also considers MPLs at some other roots of unity (for alternative proofs based on \eqref{eq:Goncoact}, see C.~Glanois \cite{Glanois:Descents}). 
For $N \in \set{2,3,4,6}$ let us define the alphabets $X_N$ by
\begin{equation}
	X_N
	\defas\begin{cases}
		\set{1,3,5,7,\ldots} &\text{if $N=2$,}\\
		\set{1,2,3,4,\ldots} &\text{if $N=3,4$ and}\\
		\set{2,3,4,5,\ldots} &\text{if $N=6$.}\\
	\end{cases}
\end{equation}
Let $X^\ast_N$ denote the set of words in letters $X_N$ (the free monoid generated by $X_N$). 
The order $1 \succ 2 \succ 3 \succ \ldots$ on $X_N$ induces a lexicographical order on $X^{\ast}_N$.
A Lyndon word is a non-empty word $w\in X_N^{\ast}$ which is inferior to each of its strict right factors, i.e.\ for all factorizations $w=uv$ with $u,v \neq 1$, we find $w<v$.
\begin{thm}[P.~Deligne \cite{Deligne:GroupeFondamental}]
\label{thm:Deligne}%
	Let $N \in \set{2,3,4,6}$, $\RU_N=\exp(2\pi\iu/N)$, and $\MZV[N]$ be the $\QQ$-algebra generated by motivic iterated integrals in the letters $\set{0, 1, \RU_N}$.
	Let $\pipow{N}$ be 2 if $N=2$ and 1 otherwise. 
	Then, the algebra $\MZV[N]$ is freely generated by the \emph{Deligne basis}
	$
		\DeligneBasis[N]
		\defas 
		\set{(2\pi\iu)^{\pipow{N}}}
			\cup
			\setexp{\Li_w(\RU_N)}{w\in X^\ast_N,w\ \text{Lyndon}}
	$.
\end{thm}
\begin{remark}
	For $N=2,6$ this is stated as Th\'{e}or\`{e}me~7.2 and Th\'{e}or\`{e}me~8.9 in~\cite{Deligne:GroupeFondamental}, but a comment is due for $N=3,4$. In these cases, the relevant Th\'{e}or\`{e}me~6.1 says that the motivic iterated integrals $\Li_w(\RU_N)$, with $w$ ranging over all words $w \in \NN^\ast=X_N^\ast$, are linearly independent over $\QQ[2\pi\iu]$.
	Therefore, the combination of the relation \eqref{eq:Li} with $2\pi\iu/N = \log(\RU_N) = I(\RU_N,0,0)$ yields an isomorphism\footnote{%
		Note that, in contrast, this map is not surjective for $N=2,6$: If $N=2$, only words with an even number of $0$'s are generated; if $N=6$, no words containing consecutive $1$'s appear.
	},
	\begin{equation*}
		(2\pi\iu)^k \Li_{n_r,\ldots, n_1}(\RU_N)
		\mapsto
		N^k 0^{\shuffle k} \shuffle \Big(
			(-1)^r\underbrace{0\ldots01}_{n_r} \ldots \underbrace{0\ldots01}_{n_1}
		\Big)
		,
	\end{equation*}
	between $\MZV[N]$ and the shuffle algebra $\QQ\langle \set{0,1} \rangle$ of words in the letters $0$ and $1$.
	This map preserves the multiplication by \axiom{shuffle} and allows us to invoke Radford's theorem~\cite{Radford:BasisShuffleAlgebra} to pick an algebra basis in terms of Lyndon words.
	It is easy to see that the Lyndon words in the $\set{0,1}$-alphabet (with respect to the order $1 \succ 0$) are mapped by \eqref{eq:Li} to the Lyndon words $w \in \NN^{\ast}=X_N^{\ast}$ as specified in Theorem~\ref{thm:Deligne}.
\end{remark}
A shuffle algebra (see \eqref{eq:shuffle}) is a polynomial algebra freely generated by its Lyndon words \cite{Radford:BasisShuffleAlgebra}.
Therefore, $\MZV[N]$ is isomorphic to $\lin_{\QQ[(2\pi\iu)^{\pipow{N}}]}X_N^\ast$.%
\footnote{%
	As isomorphism one could map $\Li_w(\RU_N)$ to the word $w$ and extend to products in $\MZV[N]$ by shuffles in $\lin_{\QQ} X^\ast_N$. The subtle construction in the next subsection has the sole aim to lift the coalgebra structure in $\MZV[N]$ to deconcatenation of words in $X^\ast_N$.
}
By counting the words in each weight, we obtain the following dimensions:
\begin{cor}\label{cor:dims}
	Let $d_{N,n}$ be the dimension of the subspaces of $\MZV[N]$ at weight $n$.
	If $N \in \set{2,6}$ then $d_{N,n}$ is the Fibonacci sequence $d_{N,n+2} = d_{N,n+1} + d_{N,n}$ with $d_{N,0}=d_{N,1}=1$.
	If $N\in\set{3,4}$ then $d_{N,n}=2^n$.
\end{cor}

\begin{remark}\label{rem:Deligne}\upshape
	In fact, Deligne proved more: If $N\in\set{2,3,4}$, then $\MZV[N]$ already contains all (motivic) iterated integrals in the letters $\set{0,\RU_N^k,k=0,\ldots,N-1}$.
In particular, $\MZV[N]$ is closed under complex conjugation (this also holds for $N=6$):
	\begin{equation}
		\MZV[N]
		= \ReTeil \MZV[N] \oplus \iu \ImTeil \MZV[N]
		= ( \MZV[N] \cap \RR ) \oplus (\MZV[N] \cap \iu\RR).
	\end{equation}
Furthermore, Deligne showed that the basis $\DeligneBasis[N]$ has minimal depth:
If we set the depth of $2\pi\iu$ to zero, then every $x \in \MZV[N]$ which is a linear combination of iterated integrals of depth $\leq d$ (at most $d$ letters are non-zero)
can be written in terms of the basis $\DeligneBasis[N]$ using only words of total depth $\leq d$. The isomorphism in the next subsection preserves depths (except for the case $N = 1$).
\end{remark}

\subsection{The $f$-alphabet}\label{sec:falphabet}%
In order to generalize \eqref{eq:falphabet-MZV}, let $\pipow{N}=2$ if $N=1,2$ and $\pipow{N}=1$ otherwise. Define the \emph{$f$-alphabets}
\begin{equation}
	F_1
	\defas \set{f_3,f_5,f_7,\ldots}
	\quad\text{and}\quad
	F_N
	\defas \setexp{f^N_k}{k \in X_N}
	\quad\text{for $N\in\set{2,3,4,6}$}.
\end{equation}
The main structure theorem for (motivic) MZVs is that for $N\in\set{1,2,3,4,6}$, the $\QQ$-algebra $\MZV[N]$ (where $N=1$ refers to ordinary MZVs) is isomorphic to \cite{Glanois:Descents,Deligne:GroupeFondamental,Brown:MMZ}
\begin{equation}
	\MZVf[N] \defas \MZVfDR[N] \otimes \QQ[(2\pi\iu)^{\pipow{N}}],
	\quad\text{where}\quad
	\MZVfDR[N] \defas \QQ\langle F_N \rangle = \bigoplus_{w\in F_N^{\ast}} \QQ w
\end{equation}
denotes the (finite) $\QQ$-linear combinations of words in the letters $F_N$.
The $\QQ$-algebra $\MZVf[N]$ is endowed with the shuffle product $\shuffle$, defined iteratively by
\begin{equation}\label{eq:shuffle}
	w \shuffle 1 = 1\shuffle w = w,
	\quad 
	av\shuffle bw=a(v\shuffle bw)+b(av\shuffle w),
\end{equation}
for words $v,w\in F_N^\ast$ and letters $a,b\in F_N$. In the notation of \axiom{shuffle} the closed formula for the shuffle product is
\begin{equation*}
	f_n\ldots f_{r+1}\shuffle f_r\ldots f_1
	=
	\sum_{\mathclap{\sigma\in\shuffles{r,n-r}}}f_{\sigma(n)}\ldots f_{\sigma(1)}.
\end{equation*}
Note that $\MZVfDR[N]$ is a free commutative algebra generated by the Lyndon words in $F_N^{\ast}$ \cite{Radford:BasisShuffleAlgebra}.
We define the \emph{weight} of $(2\pi\iu)^k f^N_{n_d}\cdots f^N_{n_1}$ as $k+n_1+\ldots+n_d$. The weight gives a grading on $\MZVf[N]$.
Furthermore, deconcatenation defines a coaction $\Delta\colon \MZVf[N] \longrightarrow \MZVfDR[N] \otimes \MZVf[N]$ of the Hopf algebra $\MZVfDR[N]$ on $\MZVf[N]$,
\begin{equation}\label{eq:coact-falphabet}%
	\Delta (w) = \sum_{w=uv} u\otimes v
	\quad\text{for words $w \in \MZVfDR[N]$ and}\quad
	\Delta(2\pi\iu) = 1\otimes2\pi\iu.
\end{equation}
Crucially, $\MZV[N] \isomorph \MZVf[N]$ are isomorphic as graded comodules, i.e.\ the weight gradings are compatible (the isomorphism as graded algebras alone already implies Corollary~\ref{cor:dims} but not the independence of the explicit generators given in Theorem~\ref{thm:Deligne}) and the formula \eqref{eq:Goncoact} for the coaction on $\MZV[N]$ simplifies to \eqref{eq:coact-falphabet} on $\MZVf[N]$. This is the reason why we present all results in the $f$-alphabet. The Galois conjugates, relevant for Conjecture~\ref{con:coaction}, are easy to read off in $\MZVf[N]$. 

An explicit construction of such an isomorphism $\psi_B\colon \MZV[N] \longrightarrow \MZVf[N]$ depends on the choice of an algebra basis $B$ of $\MZV[N]$ and is explained in detail in \cite{Brown:DecompositionMotivicMZV}. In short, the construction goes as follows:
First note that $\psi_B$ should be multiplicative, so we only need to specify $\psi_B(b)$ for basis elements $b \in B$. We always assume that $B$ contains $(2\pi\iu)^{\pipow{N}}$ and set $\psi_B( (2\pi\iu)^{\pipow{N}}) = (2\pi\iu)^{\pipow{N}}$, i.e.\ $\psi_B$ is linear over $\QQ[(2\pi\iu)^{\pipow{N}}]$.
Furthermore, for each weight $n\in X_N$ there should be precisely one \emph{primitive} basis element $b_n\in B$ of this weight. It is mapped to $\psi_B(b_n) \defas f_n^N$. By primitive we mean $\cored b_n=0$ for the twice reduced coaction (compare with \eqref{eq:coact-reduced})
\begin{equation}
	\cored x
	\defas
	\Delta x - 1 \otimes x - x^{\dR} \otimes 1,
	\label{eq:coact-twicereduced}%
\end{equation}
where $x\mapsto x^{\dR}$ denotes the projection $\MZV[N]^{\motivic} \longrightarrow \MZV[N]^{\dR}$ (or $\MZVf[N] \longrightarrow \MZVfDR[N]$) onto the quotient by $(2\pi\iu)^{\pipow{N}}$. Finally, the image $\psi_B(b)$ of all remaining (non-primitive) basis elements is recursively determined by the requirement that $\cored \psi_B (b) = (\psi_B \otimes \psi_B) \cored b$. This fixes $\psi_B(b)$ up to an element of the kernel $(2\pi\iu)^n\QQ \oplus f_n^N \QQ$ (resp.\ $(2\pi\iu)^n\QQ$ if $f_n^N$ does not exist) of $\cored$ if $b$ has weight $n$.
The convention is that the coefficients of $(2\pi\iu)^n$ and $f_n^N$ in $\psi_B(b)$ are zero.

For an arbitrary period $x\in\MZV[N]$, its image $\psi_B(x)\in\MZVf[N]$ can be computed by the \emph{decomposition algorithm} from \cite{Brown:DecompositionMotivicMZV}. It exploits the same recursion via $\cored$ and determines the coefficients of the primitives by an exact numeric algorithm.
This method is very efficient and we used a computer implementation which is part of the program \cite{Schnetz:HyperlogProcedures}.

Let us illustrate the procedure in the case of MZVs ($N=1$) where we choose for $B$ the basis given in the \emph{MZV-datamine} \cite{Datamine}.
This basis contains all odd Riemann zeta values, $\mzv{2n+1} \in B$, which are primitive by \eqref{eq:coact-singlezeta} and hence mapped to $\psi_B(\mzv{2n+1}) = f_{2n+1}$.
The first non-primitive basis element is $\mzv{5,3} \in B$. Using \eqref{eq:Goncoact} we find $\cored \mzv{5,3} = - 5\mzv[\dR]{5} \otimes \mzv{3}$, see Example~\ref{ex:53}. Applying $\psi$ to this equation shows that $\psi(\mzv{5,3}) = -5 f_5 f_3$ (the constraint from $\cored$ alone would allow us to add any rational multiple of $(2\pi\iu)^8$, but the convention is not to do so).

To lighten notation we identify elements of the $f$-alphabets $\MZVf[N]$ with their preimages in $\MZV[N]$, e.g.\ in \eqref{eq:p711}--\eqref{eq:electron} and in Section~\ref{sec:phi4}. The dependence on the basis $B$ is thereby suppressed (only $N$ can be read-off from the $f$-alphabet); our convention is to use the datamine for MZVs, Deligne's basis for $N=2$ and a modified Deligne basis (see Theorem~\ref{thm:falphabet-parity} and \eqref{eq:parity-basis-Q}) for the cases $N\geq 3$. These bases are very convenient as we will discuss in detail in Section~\ref{sec:parity}.
In Table~\ref{tab:53} we show the different $f$-alphabet expressions $\psi_B(\mzv{3})$ and $\psi_B(\mzv{5,3})$ for various bases $B$.
With the data in this table it is easy to check that the coaction on $\mzv{5,3}$ commutes with the isomorphism $\phi_B$.

\begin{remark}\upshape
	Although for $N>1$ the alphabets $F_N$ of the $f$-alphabet and $X_N$ of Deligne's basis are essentially the same, the isomorphism $\psi$ drastically changes the formula for the coaction and is hence a non-trivial map.
	An expression for periods in terms of multiple polylogarithms (Deligne's basis) can in fact become very lengthy even in cases where the $f$-alphabet seems relatively simple (compare \cite[Equation~(5.1.10)]{Panzer:PhD} with \eqref{eq:p711}).
\end{remark}

\subsection{The derivation $\delta_m$}
For some practical applications the coaction is unnecessarily complicated. In the $f$-alphabet it is clear that we obtain the same information on a word $w$ by clipping off its leftmost letter \cite{Brown:DecompositionMotivicMZV}.
\begin{defn}
	For $N\in\set{1,2,3,4,6}$ and $m\in X_N$ the linear map
$ \delta_m\colon\MZVf[N]\longrightarrow\MZVf[N] $
	is defined as follows: Given any element $w(2\pi\iu)^{k\pipow{N}}$ of $\MZVf[N]$ ($w\in F_N^\ast$, $k\in\NN_0$), set
\begin{equation}
	\delta_m\big(w(2\pi\iu)^{k\pipow{N}}\big)
	\defas \begin{cases}
		u(2\pi\iu)^{k\pipow{N}} & \text{if $w=f^N_mu$ begins with $f^N_m$, and}\\
			0 & \text{otherwise.}
	\end{cases}
\end{equation}
\end{defn}
By \eqref{eq:shuffle}, the map $\delta$ is a derivation for the shuffle product
\begin{equation}
	\delta_m(v\shuffle w)
	= (\delta_m v)\shuffle w+v\shuffle(\delta_m w).
\end{equation}
By slight abuse of notation we also write $\delta_m$ for the analogous derivation on $\MZVfDR$ as well as for $\psi^{-1}\circ\delta_m\circ\psi$ acting on $\MZV[N]$ and on $\MZV[N]^{\dR}$.

Note that the map $\delta_m\colon \MZV[N]\longrightarrow \MZV[N]$ depends on the choice of an algebra basis in $\MZV[N]$.
A formula that reduces the calculation of $\delta_m$ on weight $n$ iterated integrals ($n>m$) to lower weights has significantly less terms than the full coaction \eqref{eq:Goncoact}:
\begin{multline}\label{eq:delta}
	\delta_m I(a_{n+1},a_n\ldots a_1,a_0)
	= \\
	\sum_{k=0}^{n-m}
	(\delta_m I^{\dR}(a_{k+m+1},a_{k+m}\ldots a_{k+1},a_k))
		I(a_{n+1},a_n\ldots a_{k+m+1}a_k\ldots a_1,a_0).
\end{multline}
Note that the iterated integral on the left-hand side of the tensor product has weight $m$, so that $\delta_m$ just extracts the coefficient of $f_m^N$ in
$\psi(I^{\dR}(a_{k+m+1},a_{k+m},\ldots a_{k+1},a_k))$. Formula~\eqref{eq:delta} is hence equivalent to the formula given in \cite{Brown:DecompositionMotivicMZV}.

Conjecture~\ref{con:coaction} can be reformulated using the derivation $\delta_m$:
\begin{equation*}
	\delta_m \PhiPeriods^{\motivic}
	\subseteq\PhiPeriods^{\motivic}
	\quad\text{and}\quad
	\delta_m\PhiPeriods[n]^{\motivic}
	\subseteq\LogPeriods[n-1]^{\motivic}
	\quad\text{for all $m$}.
\end{equation*}
This point of view is used in Table~\ref{tab:periods} to verify the coaction conjecture up to eight loops.

\section{The parity basis}\label{sec:parity}
For all $N\in\set{2,3,4,6}$, the $\QQ$-algebra $\MZV$ of MZVs is a natural sub-algebra of $\MZV[N]$.
This becomes particularly important in $\phi^4$ theory where all periods up to six loops are MZVs.
In the $f$-alphabet MZVs are described by the letters $F_1$ of odd weight greater than or equal to three \eqref{eq:falphabet-MZV}. 
Consistently, letters of these weights exist in the $f$-alphabet of $\MZV[N]$ for all $N\in\set{2,3,4,6}$.
It is natural to ask if we can find algebra bases $B_N$ for $\MZV[N]$ such that under $\psi_B$, words in odd letters $f_{2n+1}^N$ with $n\geq 1$ in combination with even powers of $2\pi\iu$ correspond to MZVs: $\psi_{B_N}^{-1} (\MZVf) = \MZV$ for $\MZVf \subset \MZVf[N]$.

This is the case for Deligne's basis of alternating sums ($N=2$) but it is not the case for Deligne's bases of $\MZV[N]$ for $N\in\set{3,4,6}$. For example, we find that
\begin{equation*}
	\psi_{\DeligneBasis[6]}^{-1}(f^6_5 f^6_3)
	=-\lfrac{5}{162}\mzv{5,3}+\lfrac{18426589}{4100362560}\pi^8+\lfrac{125}{8748}\mzv{5}\pi^3\iu-\frac{950}{31}\Li_8(\xi_6)
\end{equation*}
is not an MZV. Conversely, MZVs are mapped to expressions which include odd powers of $2\pi\iu$ and even weight letters (see Table~\ref{tab:53}).
\begin{table}\centering
\renewcommand{\arraystretch}{1.2}
\begin{tabular}{r>{$}l<{$}>{$}l<{$}>{$}l<{$}}
	\toprule
	basis $B$ & \psi_B(\mzv{3}) & \psi_B(\mzv{5}) & \psi_B(\mzv{5,3}) \\
	\midrule
	datamine $\MZV$ &
		f_3
	&
		f_5
	&
		-5f_5f_3\\
	Deligne $\MZV[2]$ &
		-\tfrac{4}{3} f_3^2
	&
		-\tfrac{16}{15} f_5^2
	&
		-\lfrac{64}{9}f^2_5f^2_3+\lfrac{22319}{19391400}(\pi\iu)^8\\
	Deligne $\MZV[3]$ &
		-\tfrac{9}{4} f_3^3 - \tfrac{1}{18} (\pi\iu)^3
	&
		-\tfrac{81}{40} f_5^3 + \tfrac{1}{180} (\pi\iu)^5
	&
		-\lfrac{729}{32}f^3_5 f^3_3-\lfrac{9}{16} f^3_5(\pi\iu)^3\\
	&&&-\,\lfrac{3645}{188} f^3_8-\lfrac{131}{197400}(\pi\iu)^8\\
	Deligne $\MZV[4]$ &
		-\tfrac{32}{3} f_3^4 - \tfrac{1}{3}(\pi\iu)^3
	&
		-\tfrac{512}{15} f_5^4 + \tfrac{1}{9}(\pi\iu)^5
	&
		-\lfrac{16384}{9} f^4_5 f^4_3-\lfrac{512}{9} f^4_5(\pi\iu)^3\\
	&&&-\,\lfrac{81920}{117}f^4_8-\lfrac{521123}{1474200}(\pi\iu)^8\\
	Deligne $\MZV[6]$&
		3 f_3^6 + \tfrac{5}{54} (\pi\iu)^3
	&
		\tfrac{54}{25} f_5^6 - \tfrac{17}{2700} (\pi\iu)^5
	&
		-\lfrac{162}{5}f^6_5 f^6_3-f^6_5(\pi\iu)^3\\
	&&&-\,\lfrac{30780}{31} f^6_8+\lfrac{6265163}{42184800}(\pi\iu)^8\\
	parity $\MZV[3]$ &
		-\tfrac{9}{8} f_3^3
	&
		-\tfrac{81}{80} f_5^3
	&
		-\lfrac{729}{128}f^3_5f^3_3+\lfrac{1651}{1776600}(\pi\iu)^8\\
	parity $\MZV[4]$ &
		-\tfrac{16}{3} f_3^4
	&
		-\tfrac{256}{15} f_5^4
	&
		-\lfrac{4096}{9}f^4_5f^4_3+\lfrac{18901}{4422600}(\pi\iu)^8\\
	parity $\MZV[6]$ &
		\tfrac{3}{2} f_3^6
	&
		\tfrac{27}{25} f_5^6
	&
		-\lfrac{81}{10}f^6_5f^6_3-\lfrac{1356289}{84369600}(\pi\iu)^8\\
	\bottomrule
\end{tabular}
	\caption{The $f$-alphabet expressions of $\psi(\mzv{3})$, $\psi(\mzv{5})$, and $\psi(\mzv{5,3})$ with respect to the datamine basis \cite{Datamine} of MZVs,
 Deligne's bases $\DeligneBasis[N]$ of $\MZV[N]$ from Theorem~\ref{thm:Deligne} and the parity bases $\ParityBasis[N]$ of $\MZV[N]$ from Corollary~\ref{cor:parity-even-Deligne}.}%
	\label{tab:53}%
\end{table}

One way to obtain the desired natural embedding $\MZVf \subset \MZVf[N]$ as the image under $\psi_{B_N}$ of $\MZV \subset \MZV[N]$ is to choose a basis $B_N$ which contains an algebra basis of $\MZV$, e.g.\ the datamine basis. However, it is very difficult to explicitly write down suitable generators of $\MZV[N]$ viewed as an algebra over $\MZV$. This problem of \emph{Galois descents} is studied in \cite{Glanois:PhD,Glanois:Descents}, but only for $N=2$ an explicit algebra basis for $\MZV$ in $\MZV[N]$ has been given. Here, we take a different approach to achieve the desired embedding
$\MZVf \subset \MZVf[N]$: We exploit the parity under complex conjugation $z\mapsto \conjugate{z}$. Let us first review the situation in depth one:
\begin{lem}\label{lem:depth1-2pii}
	For $n, N \in \NN$ and all $N$\textsuperscript{th} roots of unity $\RU_N$ such that $(n,\RU_N)$ $\neq$ $(1,1)$, either the real- or the imaginary part of $\Li_n(\RU_N)$ is a rational multiple of $(2\pi\iu)^n$:
	\begin{equation}
		\Li_n(\RU_N) + (-1)^n \Li_n(\RU_N^{\ast})
		\in (2\pi\iu)^n \QQ.
		\label{eq:depth1-2pii}%
	\end{equation}
\end{lem}
\begin{proof}
This is a special case of the inversion formula \cite[Equation (7.20)]{Lewin:PolylogarithmsAssociatedFunctions},
	\begin{equation}\label{parodd}
		\Li_n(-x) + (-1)^n \Li_n(-1/x)
		= -\frac{(2\pi\iu)^n}{n!}B_n\left( \frac{1}{2} + \frac{\log x}{2\pi\iu} \right),
		\quad\text{if}\ 
		\abs{x}\leq1,
	\end{equation}
where the Bernoulli polynomials $B_n(x)$ are defined by $\frac{t e^{xt}}{e^t-1} = \sum_{n=0}^{\infty} \frac{t^n B_n(x)}{n!}$
and log denotes the principle branch of the logarithm. (An inductive proof is to verify that both sides define functions $f_n(x)$ that solve $x \partial_x f_n(x) = f_{n-1}(x)$ and to check the limits $n=1$ and $x=-1$.)
\end{proof}

\begin{lem}\label{lem:depth1-zeta}
	For $\RU_N \defas \exp(2\pi\iu/N)$ and any integer $n \geq 2$, we have the relations
\begin{align*}
	\Li_n(\RU_2)
		&= (2^{1-n}-1)\mzv{n},&
	2 \ReTeil \Li_n(\RU_3)
		&= (3^{1-n}-1)\mzv{n},\\
	2 \ReTeil \Li_n(\RU_4)
		&= (4^{1-n}-2^{1-n})\mzv{n},&
	2 \ReTeil \Li_n(\RU_6)
		&= (3^{1-n}-1)(2^{1-n}-1)\mzv{n},\\
	\ImTeil \Li_n(\RU_6)
		&= (2^{1-n}+1)\ImTeil \Li_n(\RU_3)
	.
\end{align*}
\end{lem}
\begin{proof}
	These follow from the multiplication formula \cite[Equation (7.41)]{Lewin:PolylogarithmsAssociatedFunctions}
	\begin{equation*}
		N^{1-n}	\Li_n(x^N) 
		= \sum_{k=1}^N \Li_n(x \RU_N^k),
		\quad\text{if $\abs{x}\leq 1$}.
	\end{equation*}
	For example, setting $x=\RU_6$ and $N=2$ we find $2^{1-n}\Li_n(\RU_3) = \Li_n(\RU_3^\ast) + \Li_n(\RU_6)$ which implies the last two equations from the previous ones
 (see also \cite[Equations (7.43) and (7.47)]{Lewin:PolylogarithmsAssociatedFunctions}).
\end{proof}
This observation shows that in depth one we have a very simple embedding of MZVs into $\MZV[N]$ by taking the real parts of $\Li_{n}(\RU_N)$.
In order to generalize these results to arbitrary depths, we introduce a `parity' operation.
\begin{defn}\label{def:parity}
	Let $X\subset\ZZ$ be an alphabet of integers and $\QQ\langle X\rangle$ the $\QQ$ vector space spanned by words with letters in $X$.
	We define
	\begin{equation}\label{eq:parity}
		\parity{n_d\ldots n_1}
		\defas (-1)^{n_d+\ldots+n_1+d}n_d\ldots n_1
	\end{equation}
	for words $n_d\ldots n_1$ in $X^\ast$ and extend this \emph{parity} map linearly to all of $\QQ\langle X\rangle$.
An element $w\in\QQ\langle X\rangle$ has even (odd) parity if $\parity{w}=w$ ($\parity{w}=-w$).
\end{defn}

The following parity theorem for MZVs was proved in \cite{IharaKanekoZagier:DerivationDoubleShuffle,Tsumura:CombinatorialEulerZagier,Brown:DepthGraded}:
\begin{thm}\label{thm:parity}
	If the word $n_d\ldots n_1$ has odd parity (and $n_d \geq 2$), then $\mzv{n_d,\ldots,n_1}$ is a rational linear combination of MZVs of depth $<d$:
	\begin{equation*}
		\mzv{n_d,\ldots,n_1}
		\equiv 0 \mod D_{\leq d-1} \MZV
		\quad\text{if $n_d+\ldots+n_1+d$ is odd.}
		\label{eq:MZV-parity}%
	\end{equation*}
\end{thm}
Here, we define the depth filtration as
$
	D_{\leq d} \MZV 
	\defas 
	\lin_{\QQ}\setexp{\mzv[k]{2}\mzv{n_r,\ldots,n_1}}{r \leq d}
$;
in particular, note that $\mzv{2}$ has zero depth. Moreover, the depth filtration is multiplicative: $(D_{\leq r} \MZV)\cdot(D_{\leq s}\MZV)\subseteq D_{\leq r+s} \MZV$ (analogously for $\MZV[N]$ below).

In fact, Theorem~\ref{thm:parity} also holds for multiple polylogarithms at arbitrary roots of unity \cite{Panzer:Parity}. However, for the purpose of this paper we only need the following very special case, for which we will give a self-contained proof:
\begin{prop}[Generalized parity]\label{genpar}%
	For $N \in\set{2,3,4,6}$ let $D_{\leq d} \MZV[N]$ denote the subspace of $\MZV[N]$ spanned by all $(2\pi\iu)^{k\pipow{N}}\Li_{\vect{m}}(\RU_N)$ with $\depth(\vect{m})\leq d$
	($\pipow{N}$ is 2 for $N=2$ and 1 otherwise).
	Then
	\begin{equation}
		\Li_{\vect{n}}(\RU_N^{\ast})-\Li_{\parity{\vect{n}}}(\RU_N)
		\equiv 0
		\mod D_{\leq d-1} \MZV[N]
	\end{equation}
	holds for arbitrary words $\vect{n} \in \NN^{\ast}$ (where we extend $\Li_\bullet$ by linearity).
\end{prop}
\begin{proof}
We use \eqref{eq:Li}. With \axiom{concat} we split the integration path at 1, such that
	\begin{equation*}
		I(\RU_N^{\ast},w,0)=\sum_{w=uv} I(\RU_N^{\ast},u,1)I(1,v,0)
	\end{equation*}
for the word $w = \omega_0^{n_d-1}\omega_1\!\cdots\omega_0^{n_1-1}\omega_1$. 
For terms with $u$ ending in $1$ or $v$ beginning with $1$, the iterated integrals in this formula are defined via shuffle-regularization. Explicitly, there are unique ways to write $u = \sum_k \omega_1^k \shuffle u_k$ and $v = \sum_k \omega_1^k \shuffle v_k$ such that all non-vanishing $u_k$ ($v_k$) end (begin) with $\omega_0$. Then (see e.g.\ \cite[Lemma~3.3.18]{Panzer:PhD})
	\begin{equation*}
		I(\RU_N^{\ast},u,1)
		= \sum_k \frac{\log^k \RU_N^{\ast}}{k!} I(\RU_N^{\ast},u_k,1)
		\quad\text{and}\quad
		I(1,v,0)
		= I(1,v_0,0)
		.
	\end{equation*}
	Note that every word in $u_k$ has precisely $k$ letters $\omega_1$ less than $u$ and recall that $\log \RU_N^{\ast} = -(2\pi\iu)/N$ has depth zero, hence $I(\RU_N^{\ast},u,1) \equiv I(\RU_N^{\ast},u_0,1)$ where equivalence means modulo lower depth.
	The inversion $f(z)=z^{-1}$ in \axiom{moebius} transforms $f^{\ast} \omega_0 = -\omega_0$ and $f^{\ast} \omega_1 = \omega_1 - \omega_0$ such that $f^{\ast}(u_0) \equiv \parity{u_0}$ and $I(\RU_N^{\ast},u_0,1)\equiv I(\RU_N,\parity{u_0},1) \equiv I(\RU_N,\parity{u},1)$.

	From Theorem~\ref{thm:parity} we know that $I(1,v,0)=I(1,v_0,0) \equiv I(1,\parity{v_0},0) = I(1,\parity{v},0)$ and thus conclude $I(\RU_N^{\ast},w,0) \equiv
\sum_{\parity{w}=\parity{u}\hspace{1pt}\parity{v}} I(\RU_N,\parity{u},1) I(1,\parity{v},0)= I(\RU_n,\parity{w},0)$ via \axiom{concat}.
\end{proof}
This result tells us that we only need the real (imaginary) parts of Deligne's basis elements $\Li_{\vect{n}}(\RU_{N})$ when $\vect{n}$ has even (odd) parity.
Because $2\pi\iu$ has weight $1$ and depth $0$, this is consistent with the odd parity of
\begin{equation}
	\parity{2\pi\iu} \defas -2\pi\iu.
	\label{eq:parity-ipi}%
\end{equation}
We set $\parity{\Li_{\vect{n}}}=\Li_{\parity{\vect{n}}}$ and extend by linearity.
\begin{cor}[Parity basis]\label{cor:parity-even-Deligne}
	The set	$\ParityBasis[N] \defas \setexp{b+\parity{b}^\ast}{b\in \DeligneBasis[N]}$ is an algebra basis of $\MZV[N]$ for $N\in\set{2,3,4,6}$.
\end{cor}
\begin{proof}
	We prove inductively over the depth that the algebra $\QQ[D_{\leq d}\ParityBasis[N]]$ generated by $D_{\leq d}\ParityBasis[N]\defas\setexp{w\in\ParityBasis[N]}{\depth(w)\leq d}$ contains $D_{\leq d}\MZV[N]$.
	The case $d=0$ is trivial:
	$\set{2(2\pi\iu)^{\pipow{N}}} = D_{\leq 0}\ParityBasis[N]$. 
	Then Proposition~\ref{genpar} shows (for all $b \in D_{\leq d} \DeligneBasis[N]$) that
\begin{equation*}
	2b
	= \big(b+\parity{b}^\ast\big)
	+ \big(b-\parity{b}^\ast\big)
	\in D_{\leq d} \ParityBasis[N] + D_{\leq d-1} \MZV[N]
	.
\end{equation*}
Assuming inductively that $\QQ[D_{\leq d-1} \ParityBasis[N]] \supset D_{\leq d-1} \MZV[N]$, this proves $D_{\leq d} \DeligneBasis[N] \subset \QQ[D_{\leq d} \ParityBasis[N]]$. By Remark~\ref{rem:Deligne} we thus know that $D_{\leq d}\ParityBasis[N]$ generates $D_{\leq d} \MZV[N]$.
Independence of $D_{\leq d}\ParityBasis[N]$ follows from $\#D_{\leq d}\ParityBasis[N]=\#D_{\leq d}\DeligneBasis[N]$.
\end{proof}

The proof of Corollary~\ref{cor:parity-even-Deligne} shows that $\ParityBasis[N]$ is an algebra basis of minimum depth (see Remark~\ref{rem:Deligne}).
Using Definition~\ref{def:parity} and \ref{eq:parity-ipi} we define the parity map on $\MZVf[N]$. 

Now, we can prove our main theorem:

\begin{thm}\label{thm:falphabet-parity}
	For $N \in \set{2,3,4,6}$ let $\psi \colon \MZV[N] \longrightarrow \MZVf[N]$ denote the isomorphism into the $f$-alphabet with respect to the parity basis $\ParityBasis[N]$.\footnote{In fact, the proof shows that Theorem~\ref{thm:falphabet-parity} holds with respect to an arbitrary algebra basis $B$ of $\MZV[N]$ as long as it contains only real and imaginary elements ($B \subset \RR \cup \iu\RR$).} Then
	\begin{equation}
		\psi\left( x^{\ast} \right)
		= \parity{\psi(x)}
		\quad\text{for all}\quad
		x \in \MZV[N].
		\label{eq:conjugation-under-iso}%
	\end{equation}
Furthermore, let $\MZVf \subset \MZVf[N]$ be the sub-algebra of words in letters of odd weights $\geq3$ times even powers of $2\pi\iu$. Then
\begin{equation}\label{eq:MZVU}
	\psi\left(\MZV\right) = \MZVf.
\end{equation}
\end{thm}

\begin{proof}
	We first prove \eqref{eq:conjugation-under-iso} by induction over the weight $\weight{x}$ of $x$. If $\weight{x}=0$ then the statement is trivial ($x \in \QQ$).
	Now consider a primitive $b\in \ParityBasis[N]$ of weight $n$. If $b$ is proportional to $(2\pi\iu)^{\pipow{N}}$ then \eqref{eq:conjugation-under-iso} is true by \eqref{eq:parity-ipi}.
Otherwise $b=\Li_n(\RU_N)+\Li_{\parity{n}}(\RU_N^\ast)$, so that $b^\ast=(-1)^{n+1}b$. Because $\psi(b)=f^N_n$ we have $\parity{\psi(b)}=(-1)^{n+1}\psi(b)$ and
\eqref{eq:conjugation-under-iso} holds.

For the remaining (non-primitive) basis elements $b \in \ParityBasis[N]$, we find by induction (using Sweedler's notation for the reduced coaction)
\begin{equation*}
	\cored \psi(b^\ast)
	= (\psi\otimes\psi)(\cored b^\ast)=\sum_b \psi(b_1^\ast) \otimes \psi(b_2^\ast)
	=\sum_b \parity{\psi(b_1)} \otimes \parity{\psi(b_2)}
	=\cored \parity{\psi(b)}.
\end{equation*}
This already implies $\psi(\conjugate{b}) = \parity{\psi(b)}$, because the kernel of $\cored$ is spanned by powers of $2\pi\iu$ and the letters $F_N$, all of which occur with coefficient zero in both $\psi(b)$ and $\psi(\conjugate{b})$ by the construction of $\psi$ (note $\conjugate{b}=\pm b$).
Hence, Equation~\eqref{eq:conjugation-under-iso} holds for all $x \in B$. Because parity is a homomorphism for the shuffle product, $\parity{v\shuffle w}=\parity{v}\shuffle\parity{w}$, Equation~\eqref{eq:conjugation-under-iso} extends to all $x\in\MZV[N]$.

For \eqref{eq:MZVU} it suffices to show that $\psi(\MZV)\subseteq\MZVf$. Then, the claim follows by comparing dimensions of given weights.
We again proceed by induction over weight. The weights 0 and 1 are trivial. Let $x\in\MZV$ have weight $n\geq2$.
The reduced coaction $\cored$ closes on MZVs. By induction we have $\cored\psi(x)=(\psi\otimes\psi)(\cored x)\in\MZVfDR\otimes\MZVf$, where $\MZVfDR$ is spanned by words in letters of odd weight $\geq3$. Therefore $\psi(x)\in p+\MZVf$ for some primitive $p$ of weight $n$, i.e.\  $p \in \QQ f^N_n\oplus\QQ(2\pi\iu)^n$ (if $f^N_n$ exists).
Because $x$ is real we obtain from \eqref{eq:conjugation-under-iso} that $\parity{p}=p$. If $n$ is even this implies $p\in\QQ(2\pi\iu)^n\in\MZVf$.
If $n$ is odd then $p\in\QQ f^N_n\in\MZVf$.
\end{proof}

\begin{remark}\label{rem:inclusions}
Using the coaction and Lemma~\ref{lem:depth1-zeta} we obtain the inclusions $\MZV\subset\MZV[2]\subset\MZV[4]$ and $\MZV\subset\MZV[6]\subset\MZV[3]$.
The second part of the proof of Theorem~\ref{thm:falphabet-parity} extends to these cases: $\psi_{\ParityBasis[3]}(\MZV[6])=\MZVf[6] \subset \MZVf[3]$ means that elements of $\MZVf[3]$ represent numbers in $\MZV[6]$ precisely when they are free of $f_1^3$. Analogously we have $\psi_{\ParityBasis[4]}(\MZV[2]) = \MZVf[2] \subset \MZVf[4]$.
\end{remark}
\begin{remark}\label{rem:MZVoverQ}
	Feynman periods \eqref{eq:Feynman-period} are defined over $\QQ$, see \eqref{eq:psi}. But when $N\in\set{3,6}$, numbers in $\MZV[N]$ with odd parity are imaginary
and not defined over $\QQ$. So instead of e.g.\ $2\pi\iu$, we should consider a rational multiple of
	\begin{equation*}
		\int_0^1 \frac{\dd x}{1-x+x^2}
		= \frac{2\pi}{9}\sqrt{3}.
	\end{equation*}
	Therefore, if $N=3,6$ we consider the real subspace $\ReTeil \left( \MZV[N] \otimes \QQ(\RU_N) \right)$ with the basis
	\begin{equation}
		\set{\frac{\pi}{\sqrt{3}}}
		\cup \setexp{2\ReTeil \Li_{\vect{n}}(\RU_N)}{\parity{\vect{n}}=\vect{n}}
		\cup \setexp{-2\sqrt{3} \ImTeil  \Li_{\vect{n}}(\RU_N)}{\parity{\vect{n}}=-\vect{n}}
		.
		\label{eq:parity-basis-Q}%
	\end{equation}
In Eqs.~\eqref{eq:p711}, \eqref{P711N3} and in Table~\ref{tab:periods} we conveniently used the $f$-alphabet with respect to this basis. This means that we replaced
$\sqrt{3}\iu f^3_{2n}$ and $\sqrt{3}\iu f^6_{2n}$ in the $f$-alphabet with respect to the modified Deligne basis by the letters $f^3_{2n}$ and $f^6_{2n}$. Letters with odd weight
remain unchanged.
\end{remark}

\section{Methods}\label{sec:methods}
For a long time it was not known how to calculate Feynman periods. Now we have several methods at hand which, when combined, suffice to compute all periods up to seven loops, most graphs with eight loops and several graphs with nine, ten, or eleven loops. The method has four building blocks:
\begin{enumerate}
	\item parametric integration,
	\item graphical functions,
	\item generalized single-valued hyperlogarithms and
	\item the decomposition algorithm.
\end{enumerate}
Parametric integration was developed by F.~Brown \cite{Brown:TwoPoint,Brown:PeriodsFeynmanIntegrals} and implemented by the first author \cite{Panzer:PhD,Panzer:HyperInt}. Graphical functions were defined in \cite{Schnetz:gf}.
The theory of generalized single-valued hyperlogarithms is presently developed by the second author \cite{Schnetz:GeneralizedSV} (first examples are in \cite{ChavezDuhr}).
The decomposition algorithm was suggested by F.~Brown \cite{Brown:DecompositionMotivicMZV} and implemented by the second author \cite{Schnetz:HyperlogProcedures}.

For most periods we need all four building blocks in the following order: Firstly, calculate a graphical function of low weight by parametric integration (using the representation given in \cite{GolzPanzerSchnetz:GfParam}). Secondly, derive differential equations for graphical functions of higher weights. Thirdly, solve the differential equations by single-valued integration in the space of generalized single-valued hyperlogarithms. Fourthly, use the decomposition algorithm to reduce the result to a basis of $\MZV[N]$.

For many graphs a subset of the building blocks suffice:
Lots of periods (like the zig-zag series) can be calculated with the theory of graphical functions alone \cite{BrownSchnetz:ZigZag}.

On the other hand, $P_{7,11}$ (see Figure~\ref{fig:P711}) was calculated in terms of iterated integrals using parametric integration alone \cite[Section~5.1.3]{Panzer:PhD}. The result was evaluated to \numprint{5000} significant digits and the rational coefficients for its representation in the parity basis $\ParityBasis[6]$ were provided by PSLQ \cite{FergusonBaileyArno:PSLQ}, exhausting \numprint{3000} digits.
Note that assuming Scenario~1 leads to an ansatz of the type~\eqref{eq:p711} where the $f^6_n$ are calculated by the decomposition algorithm. So in practice, the constraints from Conjecture~\ref{con:coaction} vastly reduce the dimension of the $\QQ$-vector space which (conjecturally) must contain the period. In the case of $P_{7,11}$ this means that PSLQ needed only 400 digits to identify the rational coefficients in \eqref{eq:p711}. 
Explicit representations of $P_{7,11}$ in terms of iterated integrals are given in \cite[Equation~(5.1.10)]{Panzer:PhD}, the attached files and most beautifully and concisely in \cite{Broadhurst:Aufbau}.

\begin{remark}\label{rem:period-conjecture}\upshape
	F.~Brown recently defined motivic Feynman periods $I^{\motivic}_G$ for \eqref{eq:Feynman-period} such that $\per(I^{\motivic}_G) = P(G)$ in \cite{Brown:FeynmanAmplitudesGalois} and proved the weaker version \eqref{eq:coaction-theorem} of our Conjecture~\ref{con:coaction}.\footnote{In our case of mass- and momentum free Feynman integrals, the underlying geometric construction goes back to \cite{BEK}.}
	The coaction is well-defined on $I^{\motivic}_G$, which opens up the possibility to avoid transcendence conjectures by working with motivic amplitudes throughout.

	However, our methods are currently only analytic. For example, we only know\footnote{We denote with $P_{6,2}$ both the corresponding graph $G$ and its period $P(G)$.}
	\begin{equation*}
		\per\left( I^{\motivic}_{P_{6,2}} \right)
		= P_{6,2}
		= \tfrac{1063}{9} \mzv{9} + 8 [\mzv{3}]^3
		= \per\left(  \tfrac{1063}{9} \mzv[\motivic]{9} + 8 [\mzv[\motivic]{3}]^3 \right)
	\end{equation*}
	from \cite{Schnetz:Census} and compute the coaction of $P_{6,2}$ via the right-hand side and \eqref{eq:Goncoact}, assuming the (conjectural) injectivity of $\per$. This applies to all of our results. One could avoid this conjecture if one could show
	$
		I^{\motivic}_{P_{6,2}}
		= \tfrac{1063}{9} \mzv[\motivic]{9} + 8 [\mzv[\motivic]{3}]^3
	$
	directly.

	Some of our methods (parametric and single-valued integration) should in principle admit a lift to the motivic level, but this would require considerable efforts and it is not clear how far this can be pushed in practice (considering that even the geometry of very simple graphs like $\zz{n}$ is only partially understood \cite{Doryn:ThesisCNTP,BrownDoryn:Framings}).
	
	Other methods, most notably the completion from Theorem~\ref{thm:completion}, seem to be extremely difficult to understand motivically (not even the $c_2$-invariant is known to adhere to completion, see Conjecture~\ref{con:completion}).
\end{remark}

\section{The structure of $\phi^4$ periods}\label{sec:phi4}
\subsection{Completion}
One of the fundamental properties of $\phi^4$ periods is their invariance under completion.
\begin{defn}\label{def:completion}
Let $G$ be a $\phi^4$ graph. Then, the completion $\completed{G}$ of $G$ is the $4$-regular graph (every vertex has four edges) that is obtained from $G$ by adding a vertex $\,'\infty'$ and four edges from vertices which have less than four edges to $\infty$ (see Figure~\ref{fig:zig-zag-completions}).
\end{defn}
\begin{figure}
	\centering%
	$\completed{\zz{5}}=\Graph[0.85]{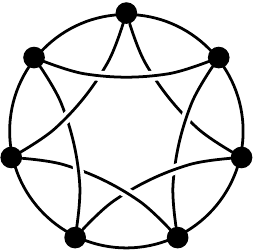}$ \qquad
	$\completed{\zz{6}}=\Graph[0.85]{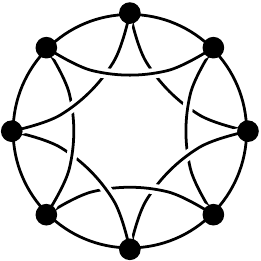}$%
	\caption{The completions of the zig-zag graphs $\zz{n}$ of Figure~\ref{fig:zig-zags} are circulants $\completed{\zz{n}}=C_{1,2}^{\,n+2}$.}%
	\label{fig:zig-zag-completions}%
\end{figure}
\begin{thm}[Definition and Theorem 2.2 in \cite{Schnetz:Census}]\label{thm:completion}%
	Let $\completed{G}$ be the completion of a {\plogdiv} $\phi^4$ graph $G$ and $v_1$, $v_2$ arbitrary vertices of $\completed{G}$.
Then $\completed{G}\setminus v_1$ and $\completed{G}\setminus v_2$ are {\plogdiv} $\phi^4$ graphs with equal period,
\begin{equation}
	P\left(\completed{G}\setminus v_1\right)
	=
	P\left(\completed{G}\setminus v_2\right).
\end{equation}
\end{thm}
A completed graph $\completed{G}$ hence represents an equivalence class of $\phi^4$ graphs with equal period.
We define $P(\completed{G}) \defas P\left(\completed{G} \setminus v  \right)$ for any vertex $v$ of $\completed{G}$. One can also define completion for non-$\phi^4$ graphs if one introduces \emph{edges with negative weights}.

\subsection{The product identity}
There exists a product formula for completed $\phi^4$ graphs. If a completed {\plogdiv} $\phi^4$ graph $\completed{G}$ can be split by removing three vertices, it is called \emph{reducible} and its period factorizes (see Figure~\ref{fig:product}).
\begin{figure}
	\centering
	$\Graph[0.7]{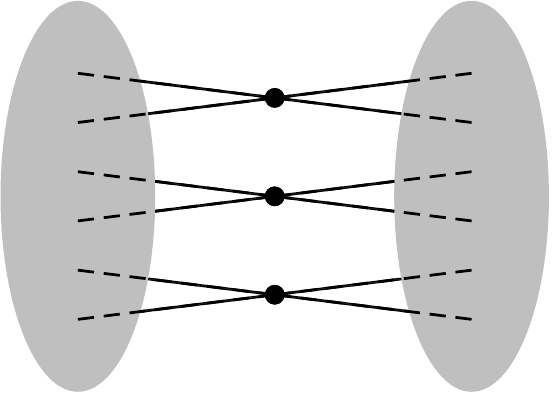} \quad = \quad \Graph[0.7]{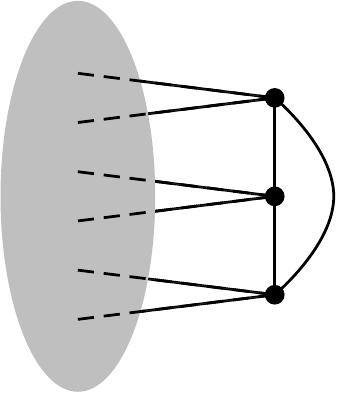} \quad \times \quad \rotatebox[origin=c]{180}{$\Graph[0.7]{factor}$}$
	\caption{A vertex cut of size three factorizes the period of a graph.}%
	\label{fig:product}%
\end{figure}
\begin{thm}[Theorem 2.1 in \cite{Schnetz:Census}]\label{prod}
Every reducible completed {\plogdiv} $\phi^4$ graph $\completed{G}$ is the gluing of two completed {\plogdiv} $\phi^4$ graphs $\completed{G_1}$ and $\completed{G_2}$ on triangle faces followed by the removal of the triangle edges. Its period is the product
\begin{equation}\label{eq:period-product}
	P(\completed{G})
	=P(\completed{G_1})P(\completed{G_2}).
\end{equation}
\end{thm}
Note that reducible graphs have weight drop, because $P(G)=P(G_1)P(G_2)$ has weight $\leq (2h_{G_1}-3) + (2h_{G_2}-3) = 2 h_G - 4 < 2 h_G-3$ due to $h_G=h_{G_1}+h_{G_2}-1$.

By the product identity it is clear that the $\ZZ$-span of periods of completed graphs with at least two edge-disjoint triangles forms a ring.%
\footnote{%
	The property of having two edge-disjoint triangles is stable under gluing on triangle faces. Graphs with only one triangle each glue to a triangle-free graph.
}
Up to seven loops the only graph which does not have two edge-disjoint triangles is the complete bipartite graph $K_{4,4}$.
The period of $K_{4,4}$ was conjectured in \cite{BK:KnotsAndNumbers} and is now proven ($P_{6,4}$ in \cite{Schnetz:Census}).
It is possible to express $P(K_{4,4})$ as integer linear combination of periods of graphs with at least two edge-disjoint triangles:
\begin{equation}
	P(K_{4,4})
	= 12P_4P_3-16P_{6,3}.
\end{equation}
Hence, the ring generated by $\phi^4$ periods up to seven loops is contained in the $\ZZ$-span of all $\phi^4$ periods.
It is unlikely that similar equations exist for the $\phi^4$ periods of the eight loop graphs without triangles and with modular $c_2$-invariants.
Still, the following questions remain open:
\begin{quest}\label{quest:alg}
	Do $\phi^4$ periods span a ring over $\ZZ$? Is $\PhiPeriods$ a (free?) $\QQ$-algebra?
\end{quest}
Because we expect that any possible obstruction that prevents $\PhiPeriods$ from being a free commutative $\QQ$-algebra appears at high loop order (compared to the loop orders considered in this article) it is presently impossible to seriously test Question~\ref{quest:alg} (see also Table~\ref{tab:dimensions} and Remark~\ref{remark:hom}).

Note that the coaction is a homomorphism with respect to multiplication,
$
	\Delta(x y) = \Delta(x)\Delta(y),
$
thus if we assume a positive answer to the second question it is sufficient to test Conjecture~\ref{con:coaction} on irreducible graphs.
In the case of generalized Feynman periods $\AllGraphPeriods$, the multiplicative structure is already proven:
\begin{thm}[{\cite[Proposition~7.9]{Brown:FeynmanAmplitudesGalois}}]
	$\AllGraphPeriods^{\motivic}$ is a $\QQ$-algebra.
\end{thm}

\subsection{$c_2$-conjectures}
The $c_2$-invariant of any connected graph with at least three vertices was defined in Definition~\ref{def:c2}.
It is conjectured that the $c_2$-invariant is an invariant of the period of the graph:
\begin{con}[Conjecture 5 in \cite{K3phi4}]\label{con:periodc2}
	If $P(G_1)=P(G_2)$ for two {\plogdiv} graphs $G_1$ and $G_2$, then $c_2(G_1)=c_2(G_2)$.
\end{con}
A weaker version of Conjecture~\ref{con:periodc2} is the following `completion conjecture'.
So far, it has been tested on all {\plogdiv} $\phi^4$ graphs up to loop order $8$.
\begin{con}[Conjecture 4 in \cite{K3phi4}]\label{con:completion}
	Graphs with identical completion have identical $c_2$.
\end{con}
By slight abuse of notation we define the $c_2$-invariant of a completed {\plogdiv} graph $\completed{G}$ as the $c_2$-invariant of $\completed{G}\setminus v$ for any vertex $v$ in $\completed{G}$ (assuming Conjecture~\ref{con:completion}).%
\footnote{%
	While the $c_2$-invariant from Definition~\ref{def:c2} makes sense also for completed graphs $\completed{G}$, in general it differs from our convention $c_2(\completed{G}) \defas c_2(\completed{G}\setminus v)$.
}
In the subsequent sections we will see how the $c_2$-invariants of {\plogdiv} graphs affect the period. The $c_2$-invariants of {\plogdiv} $\phi^4$ graphs up to ten loops are contained in the attached files.

\subsection{The ancestor}
Consider the situation that two triangles in a graph $G$ meet in an edge $e$.
A double triangle reduction of $G$ replaces a vertex of $e$ by a crossing such that a single triangle emerges, see Figure~\ref{fig:doubletriangle}.
\begin{figure}[ht]
	\centering
	$\Graph{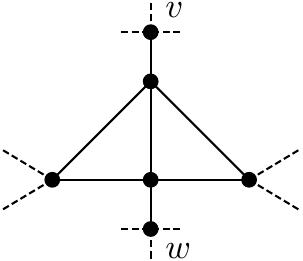}$ \qquad \scalebox{1.5}{$\Longrightarrow$} \qquad $\Graph{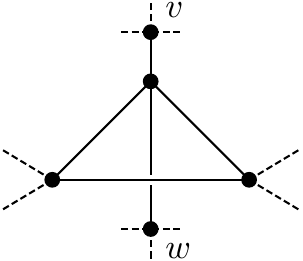}$%
	\caption{Double triangle reduction: Replace the joint vertex of two attached triangles by a crossing. Vertices $v$ and $w$ are not allowed to be identical.}%
	\label{fig:doubletriangle}%
\end{figure}
\begin{thm}[Propositions 2.2--2.4 in \cite{Schnetz:Census}]
	Any maximum sequence of product reductions (see Figure~\ref{fig:product}) and double triangle reductions (see Figure~\ref{fig:doubletriangle}) of a completed {\plogdiv} $\phi^4$ graph leads to the same graph which has completed {\plogdiv} components.
\end{thm}

\begin{defn}[Assuming Conjecture~\ref{con:completion}]\label{defanc}
The graph (possibly with several components) that is obtained from a completed {\plogdiv} $\phi^4$ graph $\completed{G}$ by a maximum sequence of product reductions and double triangle reductions is the `ancestor' $\anc(\completed{G})$ of $\completed{G}$.
The graph $\completed{G}$ is a `descendant' of $\anc(\completed{G})$.
The product of the periods of the components of $\anc(\completed{G})$ is the period $P$ of $\anc(\completed{G})$.

If $P(\anc(\completed{G}))\in\MZV[N][\RU_N]$ for some $N$ (which implies that $P(\anc(\completed{G}))$ is mixed Tate) then the `weight drop' of $\anc(\completed{G})$ is---in the case that $\anc(\completed{G})$ has one component---$2h_{\anc(\completed{G})\setminus v}-3$ minus the maximum weight of $P(\anc(\completed{G}))$ (where $v$ is any vertex in $\anc(\completed{G})$).
Factorizations increase the weight drop by 1, so that the weight drop of a general ancestor is the sum of the weight drops of its components plus the number of components minus 1.
Accordingly, the $c_2$-invariant of $\anc(\completed{G})$ is zero if $\anc(\completed{G})$ has more than one component.
Otherwise it is the $c_2$-invariant of $\anc(\completed{G})\setminus v$.
\end{defn}
For example, the zig-zag graphs in Figure~\ref{fig:zig-zag-completions} allow for the longest sequence of double triangle reductions which reduce them all the way down to $\anc(\completed{\zz{n}})=\completed{\zz{3}} = K_5$, which has weight drop $0$ because $P(\zz{3})=6\mzv{3}$ has weight $3=2\cdot 3-3 = 2 h_{\zz{3}}-3$. Its $c_2$-invariant is $c_2(\completed{\zz{3}})=c_2(\zz{3})=-1$.
In contrast, the graph $L_8=P_{8,16}$ from Figure~\ref{fig:ladders} has an ancestor $\anc(P_{8,16})=K_5^3$ with $3$ components (so $c_2(K_5^3)=0$) and therefore a weight drop of $2$ (the components $K_5$ themselves have no weight drop).
\begin{remark}
	The ancestor is (conjecturally) a refinement of the $c_2$-invariant: For any \plogdiv\ graph $G$ the $c_2$-invariant equals the $c_2$-invariant of the ancestor of its completion $\completed{G}$ in the sense of Definition~\ref{defanc},
\begin{equation}\label{c2anc}
	c_2(G)=c_2(\anc(\completed{G})).
\end{equation}
Assuming Conjecture~\ref{con:completion}, this follows from \cite[Corollary~34]{K3phi4} and \cite[Proposition~36]{BrownYeats:WD}.
\end{remark}
The ancestor predicts the maximum weight of a $\phi^4$ period:
\begin{con}\label{con:maxweight}
	Assume $P(G)\in\MZV[N][\RU_N]$ for some $N$ and a {\plogdiv} $\phi^4$ graph $G$ with completion $\completed{G}$.
	Then, the weight drop of $\anc(\completed{G})$ exists\footnote{This means $P(\anc(\completed{G})) \in \MZV[N][\RU_N]$ for some $N$, according to our Definition~\ref{defanc}. Note that the weight filtration exists on all Feynman periods \cite{Brown:FeynmanAmplitudesGalois}, so Conjecture~\ref{con:maxweight} could in principle be formulated without this restriction to periods that are MPLs. However, we have no supporting data since all known $\phi^4$ periods are MPLs.} and it equals the weight drop of $P(G)$, i.e. it equals $2h_G-3$ minus the maximum weight of $P(G)$.
\end{con}
\begin{figure}
	\centering
	\includegraphics[width=0.42\textwidth]{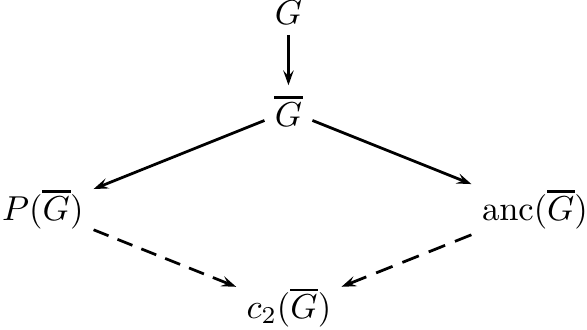}%
	\caption{The relations between a {\plogdiv} $\phi^4$ graph $G$, its completion $\completed{G}$, the period $P(\completed{G})$, the ancestor $\anc(\completed{G})$ and the $c_2$-invariant $c_2(\completed{G})$. The two dashed maps are conjectural.}%
	\label{fig:relations}%
\end{figure}
Note that two graphs with equal period may have different ancestors. The relations between the graph, its completion, the period, the ancestor and the $c_2$-invariant is depicted in Figure~\ref{fig:relations}. Going down the diagram reduces the number of objects at a given loop order. Note that the existence of the $c_2$-invariant for completed graphs and the relation to the period depend on Conjectures~\ref{con:completion} and \ref{con:periodc2}.
The ancestors of all $\phi^4$ graphs up to eleven loops are contained in the attached files.

\subsection{$c_2$-invariant $-1$}
It is conjectured in \cite{modphi4} that {\plogdiv} $\phi^4$ graphs have $c_2$-invariant $-1$ if and only if their ancestor is the complete graph with five vertices $K_5$ ($=P_3$ in \cite{Schnetz:Census}).
Their periods are conjectured to be MZVs:
\begin{con}\label{con:c2-1}
	If $G$ is a {\plogdiv} graph with $c_2$-invariant $-1$, then $P(G)$ is an MZV of weight $2h_G-3$.
\end{con}
Up to eight loops the periods of all graphs with $c_2$-invariant $-1$ are known.

\subsection{$c_2$-invariant $0$}\label{sec:c20}
Graphs with $c_2$-invariant $0$ are weight drop graphs (by definition).
If a weight drop graph $G$ has a period in $\MZV[N]$ for some $N$ then it is conjectured that the weight of $P(G)$ is $\leq 2h_G-4$.
More precisely, the results in $\phi^4$ theory are consistent with the following conjecture:
\begin{con}\label{con:weightdrop}
	Let $G$ be a \plogdiv\ $\phi^4$ graph with $c_2$-invariant $0$. 
	If $P(G)$ is mixed Tate, then either $P(G)$ is of pure weight $2h_G-4$ or $P(G)$ mixes (some) weights between $h_G+2$ and $2h_G-5$.
\end{con}
\begin{figure}
	\centering
	${L_8} = \Graph[0.9]{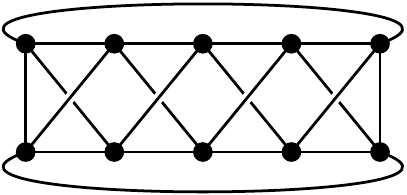}$
	\qquad
	${L_{10}} = \Graph[0.9]{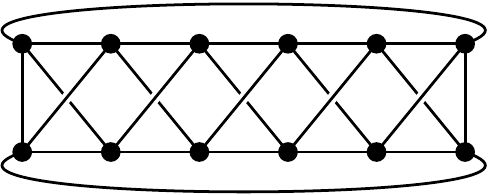}$
	\caption{The ladder $L_8=P_{8,16}$ provides the weight $10$ transcendental $Q_{10}$ in \cite{Schnetz:Census}.
	The ladder $L_{10}=P_{10,425}$ conjecturally provides the weight $12$ transcendental $Q_{12,5}$ (see Section~\ref{sec:L10}).}%
	\label{fig:ladders}%
\end{figure}
The majority of graphs with known period have $c_2$-invariants 0 or $-1$.
The first graphs with double weight drop are triple products of the (completed) graph $K_5$. They have seven loops and period $(6\mzv{3})^3$.
The first non-product graph with double weight drop is the eight loop graph $L_8=P_{8,16}$ which mixes weights $10$ and $11$ (see Figure~\ref{fig:ladders}).
It provides the weight $10$ MZV transcendental $Q_{10}$ in \cite{Schnetz:Census}. Note that $Q_{10}$ is absent at seven loops.
Because $Q_{10}\mzv{3}$ appears in several eight loop $\phi^4$ periods the existence of $P_{8,16}$ is vital for $\delta_3$ to close on $\phi^4$ periods (in Scenario~1).
We see this fact by the appearance of $P_{8,16}$ in the right column of Table~\ref{tab:periods} at the end of this article.

Note that by the above conjecture, $\phi^4$ periods in some $\MZV[N]$ of weight $\leq 11$ can only appear in $\PhiPeriods[7]$ or in multiple weight drop graphs up to nine loops.
Because $\PhiPeriods[7]$ is known and, assuming Conjecture~\ref{con:maxweight}, all multiple weight drop $\phi^4$ periods of at most nine loops are known, we conjecture that we know all mixed Tate $\phi^4$ transcendentals up to weight $11$ (see Table~\ref{tab:dimensions}).
Conjecturally, they are spanned by products of the Riemann zeta values $\mzv{3}$, $\mzv{5}$, $\mzv{7}$, $\mzv{9}$, $\mzv{11}$ and the MZVs $Q_8$, $Q_{10}$, $Q_{11,2}$ given in \cite[Tables 3a and 3b]{Schnetz:Census}, plus the period $P_{7,11}\in\sqrt{3}\ImTeil\MZV[6]$ given in \eqref{eq:p711}. Note that these elements generate a free subalgebra of $\Periods^{\motivic}$.

\begin{lem}[Follows from Scenarios 1 and 2]\label{lem:alternate}
Let $G$ be a {\plogdiv} $\phi^4$ graph of minimal loop order such that $P(G)\in\MZV[2]\setminus\MZV$ is an alternating sum but not an MZV. Then the period has weight drop, i.e.\ the weights of $P(G)$ are $\leq 2h_G-4$.
\end{lem}

\begin{proof}[Proof by contradiction]
Assume that $P(G)$ has weight $2h_G-3$ (no weight drop).
Consider the $f$-alphabet expression $\psi_{\DeligneBasis[2]}(G)$ for $P(G)$.
By Theorem~\ref{thm:falphabet-parity} it must contain a word $w$ with the letter $f^2_1$. If the leftmost letter $f^2_n$ of $w$ is not $f^2_1$, we get a contradiction: $P(G)$ would have a non-MZV Galois conjugate $\delta_nP(G)\in\MZV[2]\setminus\MZV$ ($\delta_n w$ contains the letter $f^2_1$) of lower weight. This contradicts the minimality of $P(G)$, because $\delta_nP(G)\in\PhiPeriods$ by Scenario~1.
Therefore $n=1$ and $\delta_1P(G)\neq0$. The weight of $\delta_1P(G)$ is $2h_G-4$,
but Scenario~2 implies $\delta_1P(G)\in\LogPeriods[h_G-1]$ which has maximum weight $2h_G-5$.
\end{proof}

This observation is consistent with the known data: The first alternating sum is $Q_{12,4}$ \eqref{eq:p93612} with weight $12$ in the double weight drop nine loop period $P_{9,36}=P_{9,75}$.
The single weight drop period $P_{9,108}=P_{9,111}$ is an alternating sum of weight $14$.
Assuming Scenario~1 the alternating sum $Q_{12,5}$ \eqref{eq:Q12,5} is expected in the triple weight drop graph $L_{10}$, see Section~\ref{sec:L10}.
The alternating sum $Q_{12,5}$ was found in several non-$\phi^4$ weight drop graphs with eight loops (e.g.\ in $P^\non4_{8,433}$).

Conjecture~\ref{con:weightdrop} is false for non-$\phi^4$ graphs. There exist {\plogdiv} non-$\phi^4$ graphs $G$ with $c_2$-invariant $0$ which mix all weights from 6 to $2h_G-4$ (e.g.\ the seven loop graph $P^\non4_{7,17}$). In general, compared to non-$\phi^4$ periods, $\phi^4$ periods very rarely mix weights.

\subsection{$c_2$-invariant $-z_2$}\label{sec:z2}
We were able to calculate ten $\phi^4$ periods of seven, eight, and nine loops with $c_2$-invariant $-z_2$ (see e.g.\ \cite{Schnetz:Census} for $P_{7,8}$ and \cite{Broadhurst:Bristol2011,Broadhurst:Radcor2013,Panzer:LL2014} for $P_{7,9}$).
\begin{con}\label{con:c2-z2}
	If $G$ is a {\plogdiv} graph with $c_2$-invariant $-z_2$, then $P(G)\in\MZV[2]$ of weight $2h_G-3$.
\end{con}
Rather frequently, periods of such graphs are actually MZVs. These cases include all graphs with known period and $c_2$-invariant $-z_2$ with $\leq 8$ loops; in particular $P_{7,8}$ and $P_{7,9}$.
We now give a possible explanation for this very late appearance of alternating sums (compared to sixth roots of unity, which appear starting at $7$ loops).
\begin{lem}[Follows from Scenario~2]\label{lem:delta1}
	Let $N \in \set{2,3,4}$ and $G$ be a {\plogdiv} $\phi^4$ graph with period $P(G)\in\MZV[N][\RU_N]$ with maximum weight $2h_G - 3$. 
	Then $\delta_1P(G)=0$, where $\delta_1$ is the derivation with respect to the weight one letter $f^N_1\in\MZVf[N]$ in the parity basis $\ParityBasis[N]$ of Corollary~\ref{cor:parity-even-Deligne}.
\end{lem}
\begin{proof}
By Scenario 2, $\delta_1 P(G)\in\LogPeriods[n-1]$. The maximum weight in $\LogPeriods[n-1]$ is $2(n-1)-3$ which is smaller than the weight $2n-4$ of $\delta_1 P(G)$.
\end{proof}

Assuming Scenarios~1, 2 and full knowledge of the mixed Tate $\phi^4$ transcendentals of weight $\leq 11$ (see Section~\ref{sec:c20}),
we expect the first non-MZV alternating sum among graphs with $c_2$-invariant $-z_2$ at nine loops, weight 15. Then, $\delta_3$ can give a $\QQ$-linear combination of $Q_{12,4}$, \eqref{eq:p93612}, $Q_{12,5}$, \eqref{eq:Q12,5}, and the weight 12 MZVs in $\PhiPeriods[8]$. Indeed, we found
\begin{equation}\label{delta3}
	\delta_3 P_{9,67}
	\in Q_{12,5}+\sW_{12}\PhiPeriods[8].
\end{equation}

Lemma~\ref{lem:delta1} is false for non-$\phi^4$ graphs. For example, there exists a graph $P^\non4_{8,39}$ (defined in the attached files) with $8$ loops, $c_2$-invariant $-z_2$, and a non-MZV alternating sum period such that
\begin{equation}\label{delta1Euler}
	\delta_1 P^\non4_{8,39}
	\in \frac{1}{2}Q_{12,4} + \sW_{12} \PhiPeriods[7].
\end{equation}

\subsection{$c_2$-invariant $-z_3$}
The three known $\phi^4$ periods with $c_2$-invariant $-z_3$ are $P_{7,11}$, Equation~\eqref{eq:p711}, $P_{8,33}$, and $P_{9,136}=P_{9,149}$, of loop orders seven, eight, nine, respectively.
All of them are given by numbers in $\sqrt{3} \ImTeil \MZV[6]$.
\begin{con}\label{con:z3}
	Let $G$ be a {\plogdiv} graph with $c_2$-invariant $-z_3$. Then $P(G)\in\iu\sqrt{3}(\MZV[6]\cap\iu\RR)$ with weight $2h_G-3$.
\end{con}
Because $\MZV \subset \MZV[6] \cap \RR$, this conjecture implies that periods of graphs with $c_2$-invariant $-z_3$ are never MZVs (in contrast to graphs with $c_2$-invariant $-z_2$).

Consider the derivation $\delta_m$ with respect to the modified parity basis of $\MZV[6]$ defined in Remark~\ref{rem:MZVoverQ}.
Let $P(G)$ be the period of a {\plogdiv} $\phi^4$ graph $G$ with $c_2$-invariant $-z_3$.
If $m$ is even then, by Theorem~\ref{thm:falphabet-parity}, $\delta_mP(G)\in\MZV[6]\cap\RR$. Assuming Scenario~1, $\delta_mP(G)$ is also in $\PhiPeriods$.
We conjecture that the only $\phi^4$ periods of weight $\leq11$ in $\MZV[6]\cap\RR$ are MZVs (see Section~\ref{sec:c20}).

If $m$ is odd then $\delta_mP(G)$ is a weight drop period in $\iu\sqrt3(\MZV[6]\cap\iu\RR)$.
We expect no weight drop periods in $\iu\sqrt3(\MZV[6]\cap\iu\RR)$ before weight $14$ where $\delta_3$ can give the weight 11 period $P_{7,11}\in\iu\sqrt 3(\MZV[6]\cap\iu\RR)$. This leads to the following conjecture.
\begin{con}[Follows from Scenario~1]
	Let $G$ be a {\plogdiv} $\phi^4$ graph with $c_2$-invariant $-z_3$.
Let $\delta_m$ be the derivation with respect to the modified parity basis in Remark~\ref{rem:MZVoverQ}. Then
\begin{equation}\begin{split}
	&\delta_m P(G)\in\MZV\quad\text{if $m > 2 h_G -16$ is even,}\\
	&\delta_m P(G)=0\quad\text{if $m>2h_G-17$ is odd.}
\end{split}\end{equation}
\end{con}
Because the number field $\QQ(\xi_6)$ equals the number field $\QQ(\xi_3)$ and $\MZV[6]\subsetneq\MZV[3]$, see Remark~\ref{rem:inclusions},
one might expect at some loop order to see counter-examples to Conjecture~\ref{con:z3} given by periods of $\phi^4$ graphs in $\iu\sqrt{3}((\MZV[3]\setminus\MZV[6])\cap\iu\RR)$.

Another consequence of $\MZV[6]\subset\MZV[3]$ is that one can give the known periods of graphs with $c_2$-invariant $-z_3$ in an $f$-alphabet of $\MZV[3]$.
If we use the alphabet with respect to the modified parity basis, we know by Remark~\ref{rem:inclusions} that the result is free of $f^3_1$. For $P_{7,11}$ we obtain an expression with a shorter coefficient in front of $\pi^{11}$,
\begin{equation}\label{P711N3}\begin{split}
	P_{7,11}
	&=
	-\lfrac{391190877}{43264}f^3_8f^3_3
	-\lfrac{247131}{256}f^3_6f^3_5
	-\lfrac{321489}{3328}f^3_4f^3_7
	+\lfrac{8435259}{10496}f^3_2f^3_9
	\\&\quad
	-\lfrac{229635}{64}f^3_2f^3_3f^3_3f^3_3
	+\lfrac{11494823863738427}{46501585778700}\Big(\frac{\pi}{\sqrt{3}}\Big)^{11}.
\end{split}\end{equation}
In a certain sense it is a general property that $\MZV[3]$ conversions produce smaller numbers than $\MZV[6]$ conversions. We do not think that it hints to a more fundamental connection between $\PhiPeriods$ and $\MZV[3]$. Even (much) smaller numerators and denominators are found in a basis given in \cite{Broadhurst:Aufbau}.

\subsection{$c_2$-invariant $-z_4$}
There exist one $\phi^4$ graph at eight loops ($P_{8,40}$) and three $\phi^4$ graphs at ten loops with $c_2$-invariant $-z_4$ \cite{modphi4}.
Analogously to the case of $c_2=-z_3$ our method is expected to produce maximum weight periods in $\iu(\MZV[4]\cap\iu\RR)$ which are never in $\MZV[2]$.
Therefore it is sufficient to test if a period is in principle accessible by our method without performing the actual calculations to conjecture new elements in $\PhiPeriods$.
However, our method fails for all four $\phi^4$ graphs with $c_2$-invariant $-z_4$ up to loop order ten. Interestingly, it works
(in principle) for some non-$\phi^4$ graphs of nine loops with $c_2$-invariant $-z_4$. So we can only conjecture new elements in the larger space
$\LogPeriods$ of all {\plogdiv} periods.

\subsection{Scenario~1 and the ladder $L_{10}$}\label{sec:L10}
By~\eqref{delta3}, Scenario~1 requires $Q_{12,5}\in\PhiPeriods$. Conjecture~\ref{con:weightdrop} restricts the set of possible periods to ten loops, $Q_{12,5}\in\PhiPeriods[10]$.
All $\phi^4$ weight drop periods are known up to eight loops. This rules out the possibility that $Q_{12,5}$ is in a single weight drop $\phi^4$ period at eight loops.
Because the periods of all weight drop $\phi^4$ ancestors are known up to nine loops we can use Conjecture~\ref{con:maxweight} to identify the graphs with multiple weight drop.
The periods of these graphs are known up to nine loops. Hence, $Q_{12,5}$ can only exist in a ten loop graph. The only unknown ten loop weight drop ancestor is $P_{10,1139}$.
Because all known weight drop ancestors only had single weight drop we conjecture that $Q_{12,5}$ is not in $P_{10,1139}$.
The multiple weight drop $\phi^4$ ten loop periods are known with the exception of $P_{10,47}$ and $L_{10}=P_{10,425}$.
From a partial calculation we have strong evidence that $P_{10,47}$ mixes weights 14 and 15. Conjecturally, the only possible source for $Q_{12,5}$ is therefore $L_{10}$.

The graph $L_{10}$ belongs to the family of ladder graphs $L_{2n}$ for $n\geq3$ (see Figure~\ref{fig:ladders}) which have ancestor $\anc (L_{2n}) = K_5^{n-1}$. For $n \geq 4$ these are the unique smallest non-product graphs with this ancestor.
The graph $L_8=P_{8,16}$ mixes weights $10$ and $11$ and it is the only source of $Q_{10}$ which is demanded, similarly to \eqref{delta3}, by Scenario~1.
Extrapolating from the cases $2n=6$ ($L_6=P_{6,3}$) and $2n=8$ we conjecture that for $n\geq3$ the period $L_{2n}$ mixes weights $2n+2$ to $3n-1$.

\begin{con}[Follows from Scenario~1 and Conjectures \ref{con:maxweight}, \ref{con:weightdrop}]\label{con:L10}
The period $L_{10}=P_{10,425}$ mixes weights 12 to 14. For the weight 12 part of $L_{10}$ we have
\begin{equation*}
	\sW_{12}L_{10}
	\in\QQ Q_{12,4}+\QQ^\times Q_{12,5}+\sW_{12}\PhiPeriods[8].
\end{equation*}
\end{con}
Unfortunately, none of our currently available tools allows us to calculate $L_{10}$.

\subsection{Beyond multiple polylogarithms.}
\label{sec:non-MPL}

With our methods we can only compute periods which can be expressed as linear combinations of multiple polylogarithms with algebraic coefficients and arguments. These are periods of mixed Tate motives and enjoy extra structure as compared to periods in general \cite{Brown:NotesMotivicPeriods}, including:
\begin{itemize}
	\item They have a well-defined grading by integer weights in our MZV-inspired counting (which equal half the weights from Hodge theory).
	\item The coaction acts unipotently, saying that a Galois conjugate of $x$ under the pro-unipotent part of the Galois group equals $x$ plus lower weight periods. In other words, the right-hand factors in $\Delta' x$ have weights strictly less than the weight of $x$.
\end{itemize}
For general periods, the weight is merely a filtration and we should expect to see half-integer weights. Furthermore, a period can have several Galois conjugates of the same weight.

Not a single non mixed Tate $\phi^4$ period has so far been computed, but they are known to occur starting from $8$ loops in graphs with modular $c_2$-invariants \cite{K3phi4,BrownDoryn:Framings}. Concretely, there are four $8$ loop $\phi^4$ periods with modular $c_2$-invariants ($P_{8,37}$, $P_{8,38}$, $P_{8,39}$ and $P_{8,41}$ \cite{modphi4}), and for $P_{8,37}$ it is known that the framing given by the period \eqref{eq:Feynman-period} is not of mixed Tate type \cite{BrownDoryn:Framings}.

Furthermore, the non mixed Tate contribution of $P_{8,37}$ has weight $12$. In fact, a non mixed Tate contribution to a period $P(G)$ necessarily has weight below $2 h_G-3$. It is not excluded that the other $8$-loop periods with modular $c_2$-invariant contribute non mixed Tate periods to $\PhiPeriods$ in even smaller weights. 
This is the reason why we had to restrict our statements about $\sW_{\leq 11} \PhiPeriods$ in Section~\ref{sec:c20} and Table~\ref{tab:dimensions} to the mixed Tate subspace.

In view of the possibility to find several weight $12$ Galois conjugates of $P_{8,37}$, it is unclear if our Conjecture~\ref{con:coaction} can persist beyond the mixed Tate frontier, but we have no means to probe this realm for the time being.

\section{Data}
\label{sec:data}

Two text files are attached to this article: \texttt{Periods} and \texttt{PeriodsNonPhi4}. These files contain {\Maple} readable lists of $\phi^4$ periods up to eleven loops and non-$\phi^4$ periods up to eight
loops, respectively. The data-structure in \texttt{Periods} is:
\begin{equation*}\begin{split}
	\texttt{Period} & \big[\text{loop order}, \text{number}\big]
	\defas \big[
	\text{list of graph edges}, \text{period in the $f$-alphabet},
	\\&
	\text{period as multiple polylogarithms}, \text{numerical value to $100$ digits},
	\\&
	\text{$c_2$-invariant}, \text{ancestor}, \text{size of the automorphism group} 
\big]
\end{split}\end{equation*}
Note that we list completed graphs which we introduced in Definition~\ref{def:completion}. In particular, an $\ell$-loop {\plogdiv} $\phi^4$ graph $G$ has a $4$-regular completion $\completed{G}$ with $\ell+3$ loops. For example, the entry $\texttt{Period}[3,1]$ for the only $3$-loop period $P_{3}$, known as the wheel with $3$ spokes, starts with the edge-list 
\begin{equation*}
	[\{1, 2\}, \{1, 3\}, \{1, 4\}, \{1, 5\}, \{2, 3\}, \{2, 4\}, \{2, 5\}, \{3, 4\}, \{3, 5\}, \{4, 5\}]
\end{equation*}
of the complete graph $K_5$. The only non-MZVs among our results are polylogarithms at 2nd and 6th roots of unity. The corresponding $f$-alphabet expressions refer to Deligne's basis ($N=2$) and the parity basis $\ParityBasis[6]$ of Corollary~\ref{cor:parity-even-Deligne}.\footnote{This means that powers of $\iu\sqrt{3}$ appear in the $f$-alphabet for periods with sixth roots of unity. For example, the form stored in $\texttt{Period}[7,11]$ differs from the representation given in \eqref{eq:p711}.}
In order to be absolutely clear and avoid any confusion, we also express all known periods explicitly in terms of multiple polylogarithms \eqref{eq:Li}. For non-MZVs, these are represented as 
\begin{equation*}
	\texttt{zeta}[[\xi,n_r],n_{r-1},\ldots,n_2,n_1]
	\defas
	\Li_{n_r,\ldots,n_1}(\xi)
\end{equation*}
where $\xi \in \set{-1,e^{\pm\iu\pi/3}}$ is a corresponding root of unity.
If the period of some graph is unknown, the corresponding three entries in the table are marked with \texttt{FAIL}.

In \texttt{PeriodsNonPhi4} graphs which are not in $\phi^4$ are stored with the first five entries of the above list.

In the following table we demonstrate that the known $\phi^4$ periods up to eight loops obey the coaction conjecture. For this we express the infinitesimal coaction in terms of $\phi^4$ periods.

\renewcommand{\arraystretch}{1.12}%
\begin{longtable}{>{$}l<{$}|>{$}l<{$}}
	\text{period}
	&
	\sum_mf^N_m\delta_m(P_\bullet) \\\hline
\endhead
P_1
	& 0 \\\hline
P_3
	& 6f_3P_1 \\\hline
P_4
	& 20f_5P_1 \\\hline
P_5	
	& \lfrac{441}{8}f_7P_1 \\\hline
P_{6,1}
	& 168f_9P_1 \\
P_{6,2}
	& \lfrac{2}{3}f_3P_3^2+\lfrac{1063}{9}f_9P_1 \\
P_{6,3} 
	& \lfrac{63}{5}f_3P_4-30f_5P_3 \\
P_{6,4}
	& -\lfrac{648}{5}f_3P_4+720f_5P_3 \\\hline
P_{7,1}
	& \lfrac{33759}{64}f_{11}P_1 \\
P_{7,2}
	& \lfrac{7}{12}f_3P_3P_4-\lfrac{5}{18}f_5P_3^2-\lfrac{195379}{192}f_{11}P_1 \\
P_{7,3}
	& \lfrac{1}{3}f_3P_3P_4-\lfrac{31}{9}f_5P_3^2-\lfrac{960211}{240}f_{11}P_1 \\
P_{7,4},P_{7,7} 
	& \lfrac{160}{21}f_3P_5-20f_5P_4+70f_7P_3 \\
P_{7,5},P_{7,10}
	& -\lfrac{24}{7}f_3P_5+45f_5P_4-\lfrac{63}{2}f_7P_3 \\
P_{7,6}
	& \lfrac{7}{12}f_3P_3P_4+\lfrac{145}{18}f_5P_3^2+\lfrac{502247}{64}f_{11}P_1 \\
P_{7,8}
	& f_3(7P_{6,3}-\lfrac{161}{30}P_3P_4)+\lfrac{527}{9}f_5P_3^2+\lfrac{2756439}{20}f_{11}P_1 \\
P_{7,9}
	& f_3(\lfrac{7}{2}P_{6,3}-\lfrac{133}{80}P_3P_4)-\lfrac{217}{24}f_5P_3^2+\lfrac{4136619}{160}f_{11}P_1 \\
P_{7,11}
	& f^6_2(-\lfrac{2755}{864}P_{6,1}+\lfrac{35}{27}P_3^3)+\lfrac{14}{9}f^6_4P_5+\lfrac{1017}{22}f^6_6P_4-\lfrac{36918}{43}f^6_8P_3 \\\hline
P_{8,1}
	& 1716f_{13}P_1 \\
P_{8,2}
	& f_3(\lfrac{145}{147}P_3P_5-\lfrac{27}{80}P_4^2)+\lfrac{29}{40}f_5P_3P_4+\lfrac{47}{16}f_7P_3^2+\lfrac{94871691}{22400}f_{13}P_1 \\
P_{8,3}
	& f_3(2P_4^2-\lfrac{320}{189}P_3P_5)-13466f_{13}P_1 \\
P_{8,4}
	& f_3(\lfrac{27}{80}P_4^2+\lfrac{1}{147}P_3P_5)+\lfrac{11}{40}f_5P_3P_4-\lfrac{97}{16}f_7P_3^2-\lfrac{76207221}{22400}f_{13}P_1 \\
P_{8,5}
	& \lfrac{789}{112}f_3P_{6,1}-\lfrac{2930}{147}f_5P_5+\lfrac{3549}{40}f_7P_4-180f_9P_3 \\
P_{8,6},P_{8,9}
	& \lfrac{488}{441}f_3P_3P_5-\lfrac{29}{2}f_7P_3^2-\lfrac{1717423}{336}f_{13}P_1 \\
P_{8,7},P_{8,8}
	& -\lfrac{81}{10}f_5P_3P_4+\lfrac{75}{4}f_7P_3^2-\lfrac{9819147}{2800}f_{13}P_1 \\
P_{8,10},P_{8,22}
	& \lfrac{93}{14}f_3P_{6,1}-\lfrac{1000}{49}f_5P_5+\lfrac{6993}{80}f_7P_4-\lfrac{765}{4}f_9P_3 \\
P_{8,11},P_{8,15}
	& f_3(\lfrac{2311}{504}P_{6,1}-\lfrac{2}{9}P_3^3)+\lfrac{8380}{441}f_5P_5-\lfrac{553}{8}f_7P_4+\lfrac{1171}{18}f_9P_3 \\
P_{8,12}
	& -6f_3P_{6,1}+\lfrac{2440}{49}f_5P_5+\lfrac{63}{2}f_7P_4-189f_9P_3 \\
P_{8,13},P_{8,21}
	& f_3(\lfrac{107}{63}P_3P_5-\lfrac{93}{80}P_4^2)+\lfrac{39}{40}f_5P_3P_4+\lfrac{441}{16}f_7P_3^2+\lfrac{166607569}{9600}f_{13}P_1 \\
P_{8,14}
	& f_3(\lfrac{21}{80}P_4^2+\lfrac{17}{147}P_3P_5)+\lfrac{141}{40}f_5P_3P_4-\lfrac{81}{16}f_7P_3^2+\lfrac{74218657}{67200}f_{13}P_1 \\
P_{8,16} 
	& f_3(\lfrac{3200}{49}P_5-40P_{6,3}+48P_3P_4)+f_5(-864P_4-256P_3^2) \\*
	& +2856f_7P_3-\lfrac{3670083}{5}f_{11}P_1 \\
P_{8,17},P_{8,23}
	& f_3(\lfrac{169}{5}P_{7,2}-\lfrac{1}{16}P_{8,16}-\lfrac{200399}{9207}P_{7,1}-\lfrac{21}{20}P_4^2-\lfrac{5}{12}P_3P_{6,3}+\lfrac{758}{441}P_3P_5 \\*
	& -\lfrac{43}{45}P_3^2P_4)+f_5(9P_{6,3}+\lfrac{123}{20}P_3P_4)-\lfrac{761}{24}f_7P_3^2+\lfrac{317604329}{4800}f_{13}P_1 \\
P_{8,18},P_{8,25}
	& f_3(\lfrac{727}{168}P_{6,1}+\lfrac{4}{3}P_3^3)-\lfrac{20}{3}f_5P_5-\lfrac{147}{8}f_7P_4+\lfrac{727}{6}f_9P_3 \\
P_{8,19},P_{8,27} 
	& f_3(\lfrac{235}{126}P_{6,1}-\lfrac{14}{9}P_3^3)+\lfrac{12160}{441}f_5P_5-\lfrac{91}{2}f_7P_4-\lfrac{100}{9}f_9P_3 \\
P_{8,20}
	& f_3 \big(\lfrac{1}{32}P_{8,16}+\lfrac{200399}{18414}P_{7,1}-\lfrac{169}{10}P_{7,2}+\lfrac{81}{32}P_4^2+\lfrac{5}{24}P_3P_{6,3}-\lfrac{653}{294}P_3P_5 \\*
	& +\lfrac{43}{90}P_3^2P_4 \big)+f_5(P_{6,3}+\lfrac{19}{8}P_3P_4)+\lfrac{1411}{32}f_7P_3^2+\lfrac{311697839}{44800}f_{13}P_1 \\
 P_{8,24} 
 	& f_3\big(\lfrac{35}{16}P_{8,16}+\lfrac{7013965}{9207}P_{7,1}-1183P_{7,2}+\lfrac{189}{4}P_4^2+\lfrac{175}{12}P_3P_{6,3}
	\\* &-\lfrac{152}{3}P_3P_5+\lfrac{301}{9}P_3^2P_4\big)+f_5\left(\lfrac{93}{20}P_3P_4-\lfrac{155}{2}P_{6,3}\right)+\lfrac{127}{4}f_7P_3^2
	\\* &-\lfrac{1051211241}{1400}f_{13}P_1 \\
 P_{8,26},P_{8,28} 
 	& f_3\big(\lfrac{7}{64}P_{8,16}+\lfrac{1402793}{36828}P_{7,1}-\lfrac{1183}{20}P_{7,2}+\lfrac{189}{40}P_4^2+\lfrac{35}{48}P_3P_{6,3}
	\\&-\lfrac{82}{21}P_3P_5+\lfrac{301}{180}P_3^2P_4\big)+f_5\left(\lfrac{31}{4}P_{6,3}-\lfrac{1147}{80}P_3P_4\right)+\lfrac{635}{48}f_7P_3^2
	\\&+\lfrac{303444219}{22400}f_{13}P_1 \\
 P_{8,29} 
 	& f_3\left(\lfrac{1447}{756}P_3P_5-\lfrac{91}{64}P_4^2\right)-\lfrac{899}{160}f_5P_3P_4-\lfrac{381}{64}f_7P_3^2+\lfrac{107241779}{89600}f_{13}P_1 \\
 P_{8,31}
 	& f_3\big(\lfrac{1183}{10}P_{7,2}-\lfrac{7}{32}P_{8,16}-\lfrac{1402793}{18414}P_{7,1}-\lfrac{791}{80}P_4^2-\lfrac{35}{24}P_3P_{6,3}
	\\&+\lfrac{2074}{189}P_3P_5-\lfrac{301}{90}P_3^2P_4\big)+f_5(31P_{6,3}+\lfrac{62}{5}P_3P_4)+\lfrac{127}{2}f_7P_3^2
	\\&+\lfrac{1748673539}{5600}f_{13}P_1 \\
 P_{8,32},P_{8,34} 
 	& -\lfrac{95}{7}f_3P_{6,1}-\lfrac{21600}{49}f_5P_5+1701f_7P_4+1140f_9P_3 \\
 P_{8,33} 
 	& f^6_2(-\lfrac{75052}{9207}P_{7,1}+\lfrac{68}{5}P_{7,2}-\lfrac{73}{90}P_3^2P_4)+f^6_4(-\lfrac{191}{21}P_{6,1}+2P_3^3)
	\\&+\lfrac{5184}{539}f^6_6P_5+\lfrac{156816}{1075}f^6_8P_4+\lfrac{83063999609784}{5132664845}f^6_{10}P_3\\\hline
\end{longtable}%
\captionof{table}{Known $\phi^4$ periods of graphs with at most eight loops span a comodule with respect to the Galois coaction. We chose algebra generators $P_1,P_3,P_4,P_5,P_{6,1},P_{6,3},
P_{7,1},P_{7,2},P_{8,16}$, whose products span the $\QQ$-vector spaces of $\phi^4$ MZV periods up to weight $11$. The letters $f^6_m$ refer to the modified parity basis in Remark~\ref{rem:MZVoverQ}.}%
\label{tab:periods}%

\bibliographystyle{JHEPsortdoi}
\bibliography{refs}

\end{document}